\documentclass[twoside,leqno]{article}
\usepackage[letterpaper]{geometry}
\usepackage{ltexpprt}

\usepackage{framed}		
\usepackage{microtype}
\usepackage[T1]{fontenc}
\usepackage{amssymb}
\usepackage{amsmath}
\usepackage{amsfonts}
\usepackage{mathtools}
\usepackage{thmtools}
\usepackage{thm-restate}
\usepackage{xcolor}
\usepackage{dsfont}
\usepackage{tikz}
\usepackage{xspace}
\usepackage{enumitem}
\usepackage{subcaption}
\usepackage{tikzit}
\usepackage{bm}
\usepackage{algorithm}
\usepackage{algorithmic}
\usepackage[utf8]{inputenc}
\usepackage[capitalise]{cleveref}

\newcommand{\LED}{\textsf{LED}\xspace}
\newcommand{\lcp}{\mathsf{lcp}}

\newcommand{\TDED}{$k$-\textsf{Dyck} edit distance\xspace}

\newcommand{\Dyck}{\textsf{Dyck}\xspace}
\newcommand{\ED}{\mathsf{ed_D}}
\newcommand{\eps}{\varepsilon}

\newcommand{\T}{\mathcal{T}}

\newcommand{\D}{\mathcal{D}}

\newcommand{\modulo}{\operatorname{mod}}
\newcommand{\ceil}[1]{\left\lceil{#1}\right\rceil}
\newcommand{\floor}[1]{\left\lfloor{#1}\right\rfloor}

\newcommand{\dd}{\mathinner{\ldotp\ldotp}}
\newcommand{\Oh}{O}
\newcommand{\Ohtilde}{\tilde{\Oh}}

\newcommand{\sm}{\setminus}
\newcommand{\sub}{\subseteq}
\newcommand{\colindex}{\ell}
\newcommand{\tlb}{T^{(\ell)}_b}

\newtheorem{definition}{Definition}[section]

\newtheorem{claim}[definition]{Claim}

\crefname{@theorem}{Theorem}{Theorems}
\crefname{@algorithm}{Algorithm}{Algorithms}
\crefname{fact}{Fact}{Facts}
\crefname{equation}{Eq.}{Eq.}

\newcommand*\samethanks[1][\value{footnote}]{\footnotemark[#1]}

\title{An Improved Algorithm for The $k$-Dyck Edit Distance Problem}

\author{Dvir Fried\thanks{\protect\vphantom{$2^{2^2}$}Department of Computer Science, Bar-Ilan University, Ramat Gan, Israel. Supported in part by ISF grants no.\ 1278/16 and 1926/19, by a BSF grant no.\ 2018364, and by an ERC grant MPM under the EU's Horizon 2020 Research and Innovation Programme (grant no. 683064).} \and Shay Golan\samethanks[1] \thanks{University of California, 
		Berkeley, USA. Partly supported by Fulbright Postdoctoral Fellowship.} \and Tomasz Kociumaka\thanks{University of California, 
Berkeley, USA. Partly supported by NSF 1652303, 1909046, and HDR TRIPODS 1934846 grants, and an Alfred P. Sloan Fellowship.} \and
Tsvi Kopelowitz\samethanks[1] \and Ely Porat\samethanks[1] \and Tatiana Starikovskaya\thanks{DI/ENS, PSL Research University, France. Partially supported by the grant ANR-20-CE48-0001 from the French National Research Agency (ANR).}}

\date{}
\begin{document}

\maketitle

\begin{abstract}\small\baselineskip=9pt\lineskiplimit=-2pt
A \Dyck{} sequence is a sequence of opening and closing parentheses (of various types) that is balanced.
The \Dyck{} edit distance of a given sequence of parentheses $S$ is the smallest number of edit operations (insertions, deletions, and substitutions) needed to transform $S$ into a \Dyck{} sequence.
We consider the threshold \Dyck{} edit distance problem, where the input is a sequence of parentheses $S$ and a positive integer $k$, and the goal is to compute the \Dyck{} edit distance of $S$ only if the distance is at most $k$, and otherwise report that the distance is larger than $k$.
Backurs and Onak [PODS'16] showed that the threshold \Dyck{} edit distance problem can be solved in $O(n+k^{16})$ time.

In this work, we design new algorithms for the threshold \Dyck{} edit distance problem which costs $O(n+k^{4.544184})$ time with high probability or $O(n+k^{4.853059})$ deterministically.
Our algorithms combine several new structural properties of the \Dyck{} edit distance problem,
a refined algorithm for fast $(\min,+)$ matrix product, and a careful modification of ideas used in Valiant's parsing algorithm.
\end{abstract}

\section{Introduction}\label{sec:intro}
The notion of balanced sequences of parentheses is a fundamental concept in the theory of computer science.
Formally, a \emph{$\Dyck$ sequence} is a sequence of opening and closing parentheses (of various types) that is balanced, and $\Dyck(t)$ is the language containing all balanced sequences of parentheses  over $\Sigma=\{(_i,)_i\mid i\in[0\dd t-1]\}$.
The $\Dyck(t)$ language is a context-free language that is generated by the following grammar: $S \to \eps$, where $\eps$ is the empty string, $S \to (_0 S )_0 S$, $S \to (_1 S )_1 S,\ldots,S \to (_{t-1} S )_{t-1} S$.

The $\Dyck$ language has proven to be instrumental in the theory of context-free languages; in particular, Chomsky and Schotzenberger~\cite{CS63} proved that every context-free language can be mapped to a subset of the $\Dyck$ language (see also~\cite{Harrison78,10.5555/549365}).
The $\Dyck$ language also plays a role in practice since balanced parentheses are used to succinctly describe arbitrary rooted trees~\cite{MR01}, and many programming languages have a balanced parentheses-like structure.
Additionally, data files, such as XML files, which store data in a structured way, often utilize a notion of balanced parentheses.
Moreover, the $\Dyck$ language is strongly related to DNA/RNA sequences since these sequences slightly deviate from the balancedness requirement, so understanding the behaviour of the $\Dyck$ language is often an important building block for designing algorithms on such sequences~\cite{GUTELL2000335,10.1007/978-3-642-01551-9_14}.

\subsubsection*{Language edit distance.}
The \emph{language edit distance} (\LED) problem is one of the most fundamental problems in formal language theory, generalizing both parsing and string edit distance problems. Formally, given a language $L$ and a string $S$, the task is to compute the minimal edit distance between $S$ and the strings in $L$.

The \LED{} problem has been extensively studied for context-free languages.
Consider a context-free language generated by a grammar $G$ and a string $S$ of length $n$.
Aho and Peterson~\cite{doi:10.1137/0201022} showed a dynamic-programming  algorithm with runtime $\Oh(|G|^2 n^3)$.
Myers~\cite{MYERS199585} improved the running time of the algorithm to $\Oh(|G| n^3)$.
In a recent breakthrough paper, Bringmann et al.~\cite{BGSW19} bypassed the $n^3$ barrier and demonstrated $\Oh(|G|^{\Oh(1)} n^{2.8244})$-time randomized and $\Oh(|G|^{\Oh(1)} n^{2.8603})$-time deterministic  algorithms.
The algorithms of~\cite{BGSW19} are non-combinatorial, in the sense that they use fast matrix multiplication (FMM).
Unfortunately, this seems to be unavoidable: the lower bound result of Lee~\cite{10.1145/505241.505242} implies that there is no algorithm that solves
\LED{} in time less than that of Boolean matrix multiplication.
Saha studied the problem of approximating \LED{}, developing a $(1+\eps)$-factor approximation algorithm~\cite{7354391}
followed by a solution providing an additive approximation of \LED{}~\cite{8104067}. In another related work, Jayaram and Saha~\cite{jayaram_et_al:LIPIcs:2017:7454} studied \LED{} for the class of linear grammars.

\subsubsection*{\Dyck{} edit distance.}
In this work, we focus on the \LED{} problem for the \Dyck{} language.
Formally, the \Dyck{} edit distance of a given sequence of parentheses $S = S[0]S[1]S[2]\cdots S[n-1]\in \Sigma^n$, denoted by $\ED(S)$, is the smallest number of edit operations (insertions, deletions, and substitutions) needed to transform $S$ into a \Dyck{} sequence.
Notice that it is enough to consider only deletions and substitutions since insertions can be transformed into deletions.

One might hope that the \LED{} problem for the $\Dyck$ language might be easier than for general context-free languages, but as Abboud et al.~\cite{DBLP:conf/focs/AbboudBW15a} showed, this is probably not the case: they proved that
an efficient algorithm for the \Dyck{} edit distance problem implies an algorithm for  $k$-Clique whose runtime is faster than what is believed to be obtainable.
In other words, an efficient algorithm for the \Dyck{} edit distance problem must either provide approximate answers or focus on \textbf{the small distance regime.}
Saha~\cite{6979046} presented a randomized algorithm with running time $\Oh(n^{1+o(1)})$ and a $\Ohtilde(1)$ multiplicative approximation factor. Very recently, Debarati, Kociumaka, and Saha~\cite{ApproxDyckED} showed a constant factor approximation algorithm with runtime $\Ohtilde(n^{1.971})$ and a $(1+\eps)$-approximation algorithm with runtime $\Ohtilde(n^2 / \eps)$.
\emph{In this work, we turn to the small distance regime. }

\subsection{The $k$-Dyck edit distance problem.}
For a string $S \in \Sigma^n$, define $\ED(S)$ to be the \emph{edit cost}  of $S$ which is the smallest number of deletions and substitutions of a parenthesis such that the resulting string belongs to the \Dyck{} language.
In this paper, we consider the \emph{\TDED{} problem}, where the input is a sequence of parentheses $S\in \Sigma^n$ and a positive integer $k$,
and the goal is to compute $\min\{\ED(S),k+1\}$.
In other words, we want to compute the $\Dyck$ edit distance of $S$, as long as the distance is at most $k$ (and return $k+1$ otherwise).
This problem was first considered by Backurs and Onak~\cite{BO16}, who gave an $\Oh(n+k^{16})$-time algorithm. Krebs et al.~\cite{10.1007/978-3-642-22993-0_38} presented a randomized streaming algorithm for the \Dyck{} language with two types of parentheses with running time $\tilde\Oh(nk^{\Oh(1)})$ and space complexity $\Oh(k^{1+\varepsilon} + \sqrt{n \log n})$, for any fixed constant $\varepsilon > 0$. Debarati, Kociumaka, and Saha~\cite{ApproxDyckED} showed a $(3+\eps)$-approximation algorithm for the \emph{\TDED{} problem} with runtime $\Ohtilde(kn/\eps)$.

\subsubsection*{$\Oh(n^3)$-time dynamic programming algorithm.} We begin by describing a simple folklore dynamic-programming algorithm that solves the $\Dyck$ edit distance problem in $\Oh(n^3)$ time. This algorithm forms the basis of previous results, as well as the basis of our new algorithms for the \TDED{} problem.

Consider a sequence $S \in \Sigma^n$ and let $\D$ be a matrix of size   $(n+1) \times (n+1)$ such that $\D[i,j]=\ED(S[i\dd j))$ (where $S[i \dd j)$  is the substring of $S$ from position $i$ to position $j-1$) to the \Dyck{} language. By definition, $\ED(S) = \D[0, n]$. Notice that $\D[i, i] = 0$ for $i\in [0\dd n]$, $\D[i,i+1]=1$ for $i\in [0\dd n)$, and $\D[i,j]$ satisfies the following recursion for $i,j\in [0\dd n]$ with $j-i\ge2$:\footnote{\vphantom{$2^{2^2}$}For $i,j\in \mathbb{Z}$, we denote $[i\dd j]=\{k\in \mathbb{Z} : i \le k \le j\}$, $[i\dd j)=\{k\in \mathbb{Z} : i \le k < j\}$, $(i\dd j]=\{k\in \mathbb{Z} : i < k \le j\}$, and $(i\dd j)=\{k\in \mathbb{Z} : i < k < j\}$.}
\begin{equation}\label{eq:basic_recursion}
	\D[i,j] = \min
	\begin{cases}
		\D[i, m] + \D[m, j] \qquad \text{for } m\in (i\dd j),\\
		\ED(S[i]S[j-1]) + \D[i+1,j-1]. &
	\end{cases}
\end{equation}
One can use \cref{eq:basic_recursion} to compute $\D[0,n]$ in $O(n^3)$ time.

\subsubsection*{\boldmath$\Oh(n+k^{16})$-time algorithm by Backurs and Onak~\cite{BO16}.}
Backurs and Onak~\cite{BO16} solved the \TDED{} problem by first showing that  the problem can be reduced to an instance on a string $S'$ of length $n' \le n$ such that: (1) $\ED(S) = \ED(S')$, and (2) when solving \TDED{} on $S'$, it suffices to compute only $\Oh(k^{14})$ values of $\D$ when applied to $S'$.
Moreover, by modifying the algorithm of  Landau, Myers and  Schmidt~\cite{LMS98}, after $\Oh(n)$-time preprocessing,
each of the $O(k^{14})$ calls to $\D$ is computed recursively in $\Oh(k^2)$ time.

\subsubsection*{Our results.}
We begin by designing an algorithm for \TDED{} which costs $\Oh(n^2k)$ time, and is thus faster than the algorithm of Backurs and Onak~\cite{BO16} whenever $k = \Omega(n^{2/15})$.
We further improve on the algorithm of Backurs and Onak~\cite{BO16} by designing a combinatorial algorithm for \TDED{} which runs in $O(n+k^5)$ time.
Our final contribution, summarized in \cref{thm:main}, improves upon this runtime using FMM.

\begin{theorem}\label{thm:main}
		Given a sequence $S$ of parentheses of length $n$, $\min\{\ED(S), k+1\}$ can be computed in $\Oh(n+k^{4.853059})$ deterministically or in $\Oh(n+k^{4.544184})$ time with high probability.
\end{theorem}

If $k>\sqrt{n}$, our algorithm is actually faster than claimed in \cref{thm:main}. The precise bound, which never exceeds $\Oh(n^2k)$, is provided in \cref{sec:overview}; it depends on the complexity of rectangular matrix multiplication and involves optimization over auxiliary parameters of the algorithm.

\subsection{Algorithmic Overview and Organization.}\label{sec:overview}

\subsubsection*{Heights and valleys.}
In \cref{sec:comb_algo}, we describe an $\Oh(n^2k)$-time algorithm for the \TDED{} problem which is based on \cref{eq:basic_recursion}. The main idea is to reduce the number of options for $m$ in \cref{eq:basic_recursion} from $\Oh(n)$ to $\Oh(k)$.
To do so, we define the \emph{height} of position $i\in [0\dd n]$ in $S$, denoted by $H(i)$, to be the difference between the number of opening parentheses and the number of closing parentheses in $S[0\dd i)$. (See \cref{fig:heights}; we remark that our notion of height differs slightly from the notion of height in~\cite{BO16}.)
The notion of heights leads to a set of positions $V$, called \emph{valleys} and defined as positions $i$ such that $H(i-1)>H(i)<H(i+1)$.
The result of Backurs and Onak~\cite{BO16} implies that one can convert $S$ into a string $S'$ of length at most $n$ such that $\ED(S) = \ED(S')$ and $S'$ has at most $2k$ valleys. The intuition is that if $\ED(S[i-1]S[i])=0$ for some peak $i$ (i.e., a position such that $H(i-1)<H(i)>H(i+1)$), then the two parentheses can be greedily matched; otherwise, they contribute at least $1$ to $\ED(S)$.
Thus, we assume without loss of generality that $S$ has at most $2k$ valleys.
The main new insight behind our $\Oh(n^2k)$-time algorithm is that one can reduce the set of options for $m$ in \cref{eq:basic_recursion} to the set $(i \dd j)\cap (M\cup\{i+1,i+2,j-2,j-1\})$, where  $M = \bigcup_{v \in V} \{v-1, v, v+1\}$; see \cref{lm:mid_point}. Since $|V|\le 2k$, this reduces the cost of each recursive call to $\Oh(k)$ time.

\begin{center}
\begin{figure}[h]
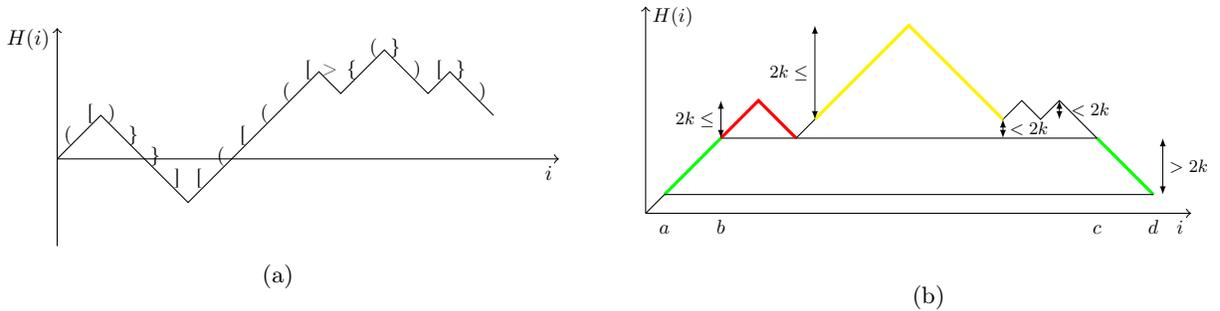
%
\begin{subfigure}{0.5\textwidth}
	\centering
	\ctikzfig{heights}
	\caption{}\label{fig:heights}
\end{subfigure}%
\begin{subfigure}{0.5\textwidth}
	\centering
	\ctikzfig{trapezoids}
	\caption{}\label{fig:trapezoids}
\end{subfigure}%

	\caption{ (a) A plot example of the function $H(\cdot)$. (b) An example of a maximal tall trapezoid $(a,b,c,d)$, which covers the green indices. Notice that the red and yellow indices also form tall maximal trapezoids (a triangle is a degenerated instance of a trapezoid).}

\end{figure}
\end{center}

\subsubsection*{Trapezoids.}
In order to obtain our $\Oh(n+k^5)$ time algorithm, we introduce the notion of \emph{trapezoids} and \emph{clusters}, which are based on the heights of positions.
Intuitively, a trapezoid is a quadruple $(a,b,c,d)$ with $0\le a < b \le c < d \le n$, where $S[a\dd b)$ are all open parentheses, $S[c\dd d)$ are all closing parentheses, $H(a) = H(d)$, $H(b)=H(c)$, and for all $i \in (b\dd c)$ we have  $H(i)\ge H(b)$; see \cref{fig:trapezoids}.

The motivation for defining trapezoids is that, in \cref{sec:alg}, we design an $\Oh(k^2)$-time algorithm for processing a trapezoid that, given values
$\min\{\D[i,j], k+1\}$ for enough $i,j$ in the vicinity of $b$ and $c$,
outputs the values $\min\{\D[i,j], k+1\}$ for enough $i,j$ in the vicinity of $a$ and $d$.
Our algorithm utilizes the approach developed by Landau--Vishkin for computing the threshold edit distance~\cite{LV97,LMS98}.
Thus, the trapezoids together with the trapezoid processing algorithm provide a method for shortcutting the evaluation of \cref{eq:basic_recursion}.
In order to leverage these shortcuts, and since the runtime per trapezoid does not depend on the distance between $d$ and $a$, we choose to apply the trapezoid algorithm to the largest possible trapezoids.
Thus, we define a \emph{maximal trapezoid} to be a trapezoid $(a,b,c,d)$ that cannot be \emph{extended} to another trapezoid of the form either $(a-1,b,c,d+1)$ or $(a,b+1,c-1,d)$, and consider only \emph{tall} maximal trapezoids which are maximal trapezoids where $b-a \ge 2k$ (and so the shortcutting is significant); see \cref{fig:trapezoids}.

\begin{figure}[ht]
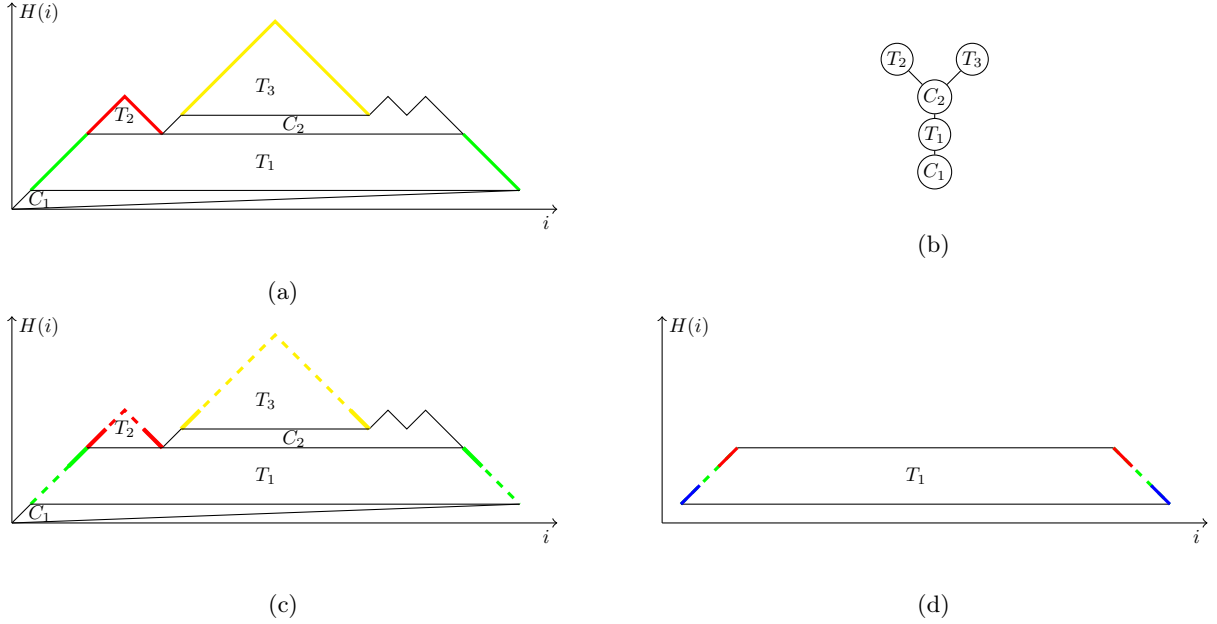
%
\begin{subfigure}{0.5\textwidth}
	\centering
	\ctikzfig{trapezoidsandclusters}
	\caption{}\label{fig:trapezoids-and-clusters}
\end{subfigure}%
\begin{subfigure}{0.5\textwidth}
	\centering
	\ctikzfig{tree}
	\caption{}\label{fig:the-tree}
\end{subfigure}\\
\begin{subfigure}{0.5\textwidth}
		\centering
	\ctikzfig{cluster}
	\caption{}\label{fig:cluster-inputs-output}
\end{subfigure}%
	\begin{subfigure}{0.5\textwidth}
	\centering
	\ctikzfig{trapezoid}
	\caption{}\label{fig:trapezoid-input-output}
\end{subfigure}%

\caption{(a) An example of the partitioning to trapezoids and clusters.
(b) The tree $\T$ for the partitioning of \cref{fig:trapezoids-and-clusters}.
(c) Focus on a cluster $C_2$, the input contains the values $\min\{\D[i,j],k+1\}$ for all pairs of red solid indices and for all pairs of yellow solid indices.
The output of the processing of $C_2$ is the  values $\min\{\D[i,j],k+1\}$ for all pairs of green solid indices.
(d) The values $\min\{\D[i,j],k+1\}$ for pairs of red solid indices are the input for processing a trapezoid $T_1$, and the values $\min\{\D[i,j],k+1\}$ for pairs of blue solid indices are the output of the processing.
}
\end{figure}

\subsubsection*{Clusters.}
The computation of \cref{eq:basic_recursion} is not fully completed by just processing all tall maximal trapezoids. In order to address the remaining computation, we introduce the notion of \emph{clusters}, which are defined as follows.
Consider a cycle with vertices $[0\dd n]$ and edges $\{(i,i+1) : i\in [0\dd n)\}\cup\{(n,0)\}$. For each tall maximal trapezoid $(a,b,c,d)$, delete edges $(i,i+1)$ for $i\in [a\dd b)\cup [c\dd d)$,
delete vertices $(a\dd b)\cup (c\dd d)$, and add edges~$(a,d)$ and $(b,c)$. The connected components of the resulting graph are called \emph{clusters}; see~\cref{fig:trapezoids-and-clusters}.

From the way that clusters are constructed, the clusters and tall maximal trapezoids define a parent-child relationship as follows. 	
Let $T=(a,b,c,d)$ be a tall maximal trapezoid and let $C$ be a cluster.
We say that $C$ is the parent of $T$ if $a,d\in C$,
and that $T$ is the parent of $C$ if $b,c\in C$.
This relationship defines a tree $\T$ whose root is the cluster containing both $0$ and $n$.
Notice that every trapezoid in $\T$ has a cluster parent and exactly 1 cluster child. Moreover, every cluster in $\T$ has an arbitrary number of trapezoid children, and at most 1 trapezoid parent (the root of $\T$ is a cluster and has no parent); see~\cref{fig:the-tree}.

By combining ideas based on the notions of heights and valleys, we design an algorithm for processing a cluster $C$ with $r$ children $T_{q} = (a_{q},b_{q},c_{q},d_{q})$ for $q\in [1\dd r]$ and a parent $T = (a,b,c,d)$, which given the values $\min\{\D[i,j], k+1\}$ for $i,j$ in the vicinity of $a_{q}$ and $d_q$ for all $q\in [1\dd r]$ (which are the outputs of the executions of the trapezoid processing algorithm on $T_q$ for $q\in [1\dd r]$), computes
the values $\min\{\D[i,j], k+1\}$ for $i,j$ in the vicinity of $b$ and $c$ (which can later be used as input for processing $T$).
Moreover, if $C$ is the root of $\T$, then the cluster processing algorithm computes $\D[0,n]$.
Similar to the $\Oh(n^2k)$-time algorithm, the intuition behind the cluster processing algorithm is to reduce the options for $m$ when applying \cref{eq:basic_recursion} to the vicinity of the valleys contained in the cluster; see~\cref{fig:cluster-inputs-output,fig:trapezoid-input-output}.

\subsubsection*{Combining trapezoids and clusters.}
The tree $\T$ provides a natural processing order, starting from the leaves (which could be either clusters or trapezoids) towards the root, in order to compute $\D[0,n]$ at the root.
Whenever the processing order reaches a trapezoid, the input to the algorithm for processing the trapezoid is either trivial when the trapezoid is a leaf, or given from the cluster child of the trapezoid.
Whenever the processing order reaches a cluster, the input to the algorithm for processing the cluster is either the output of the trapezoid processing algorithm when applied to the child trapezoid of the cluster, or trivial when the cluster is a leaf.
Finally, the output  $\D[0,n]$ is computed when processing the cluster that is the root of $\T$.

In \cref{sec:alg}, we show that the total runtime of the cluster processing algorithm on all clusters is $\Oh(k^5)$.
We also show that the number of tall maximal trapezoids is $\Oh(k)$ and so the total time spent processing trapezoids is $\Oh(k^3)$. In addition, constructing the tall maximal trapezoids, the clusters, and the tree $\T$, together with some additional preprocessing,  costs $\Oh(n)$ time, for a total of $\Oh(n+k^5)$ time.

\subsubsection*{Min-plus products, fast matrix multiplication, and Valiant's parser.}
Bringmann et al.~\cite{BGSW19} improved the running time of the dynamic programming algorithm based on \cref{eq:basic_recursion}
from $\Oh(n^3)$ to $\Oh(n^{2.8603})$ (deterministically) and $\Oh(n^{2.8244})$ (using randomization).
Their approach is based on observations that $\D$ satisfies a \emph{column bounded difference (column-BD) property}\footnote{\vphantom{$2^{2^2}$}The notion of bounded difference defined in~\cite{BGSW19} is parametrized by a natural number $W$.
In our setting $W=1$, so we only focus on this case.} $|\D[i,j]-\D[i+1,j]|\le 1$ and a \emph{row bounded difference (row-BD) property} $|\D[i,j-1]-\D[i,j]|\le 1$
all $i,j\in [0\dd n]$ with $i>j$.
Moreover,  the bottleneck of computing $\D[i,j]$ is
the $(\min,+)$ product of vectors $(\D[i,m])_{m=i+1}^{j-1}$ and $(\D[m,j])_{m=i+1}^{j-1}$.
If a matrix is both column-BD and row-BD, we say that the matrix is \emph{fully-BD}.
Their main technical contribution for the $k$-Dyck edit distance problem is an efficient algorithm for computing the $(\min,+)$ product\footnote{For two matrices $A$ and $B$ of size $a\times b$ and $b\times c$, respectively, the $(\min,+)$-product of $A$ and $B$ is a matrix $C=A\star B$ of size $a\times c$ such that $C_{i,j}=\min_{\colindex\in [1 \dd b]}\{A_{i,\colindex}+B_{\colindex,j}\}$.
} of  fully-BD square matrices. Building upon Valiant's parser, filling the entire dynamic-programming table $\D$ is then reduced to
computing the $(\min,+)$ products of the already constructed square submatrices of $\D$.
It is worth mentioning that Bringmann et al.~\cite{BGSW19} introduce a subcubic $(\min,+)$ algorithm also for the case where only one of the matrices is either column-BD or row-BD.\@
A subcubic $(\min,+)$ algorithm for even less structured matrices was later introduced by Vassilevska-Williams and Xu~\cite{VWX20}.
	Recently, Chi et al.~\cite{CDX22} improved upon the result of Bringmann et al.~\cite{BGSW19} for bounded difference $(\min,+)$ product, and, in follow up work, Chi et al.~\cite{CDXZ22} introduced even more efficient algorithms that apply to a broader range of structured matrices. 
All of these recent improvements are randomized. 

In the dynamic programming based on limiting the options for $m$, which results in the $\Oh(n^2k)$ time algorithm, the bottleneck is determining the $(\min,+)$ product of vectors $(\D[i,m])_{m\in (i\dd j)\cap M}$ and $(\D[m,j])_{m\in (i\dd j)\cap M}$.
Thus, instead of multiplying square matrices, it suffices to multiply smaller rectangular matrices (with columns of the first matrix and rows of the second matrix restricted to $M$).
While this shrinks the two matrices, only the column-BD property is preserved for the left matrix in the $(\min,+)$ product, while only the row-BD property is preserved for the right matrix.
Using the $(\min,+)$ product algorithm of Bringmann et al.~\cite{BGSW19} for the more general case, where only one of the matrices is either column-BD or row-BD would lead to inferior running times.

Nevertheless, in \cref{sec:min-plus-rect}, we observe that the output matrix of the $(\min,+)$ product in our case is fully-BD and, accordingly, design a $(\min,+)$ product algorithm specialized for our case, which may be of independent interest.
Formally, we prove the following theorem in \cref{sec:min-plus-rect}.

\begin{restatable}{theorem}{rectmpp}\label{thm:rectangular-min-plus}
	Let $A$ and $B$ be  integer matrices of sizes $n\times n^\alpha$  and $n^\alpha\times n$, respectively, such that $A$ is column-BD and $B$ is row-BD.\@
	For any $\beta,s,\delta\in (0,1]$, there exists a randomized algorithm that computes the $(\min,+)$ product $A\star B$ (correctly with high probability) in time
	\[\Ohtilde\left(n^{\beta} M(n,n^\alpha,n)+n^{2.5-\beta}\right)\]
	and a deterministic algorithm that runs in time
\[\Ohtilde\left(n^{\delta+s} M(n,n^{\alpha-s/3},n)+n^{\alpha}\cdot M(n^{1-\delta},n^s,n^{1-\delta})+n^{2+\alpha-s/3} + n^{\alpha}\cdot M(n^{1-\delta},n^{1-\delta},n^{1-\delta})\right).\]
Here, $M(a,b,c)$ denotes the time needed for computing the Boolean product of an $a\times b$ matrix with a $b\times c$ matrix.
\end{restatable}

Henceforth, we assume that the running time of \cref{thm:rectangular-min-plus}, after optimizing over $\beta,s$ and $\delta$, is $\Ohtilde(n^{\bar\omega(\alpha)})$ (depending on context, this can either refer to the deterministic or the randomized runtime).
In \cref{sec:min-plus-rect}, we estimate  $\bar{\omega}(\frac12)$ (using the bounds of~\cite{LGU18,AVW21}),
deriving $\bar{\omega}(\frac12)<2.426524$ in the deterministic case and $\bar{\omega}(\frac12)<2.272092$ if randomization is allowed (for this case, a similar result was obtained independently and in parallel by~\cite{https://doi.org/10.48550/arxiv.2208.02862}). 

We emphasize that the techniques used to prove \cref{thm:rectangular-min-plus} mostly follow the paradigm and techniques of Bringmann et al.~\cite{BGSW19} and Chi et al.~\cite{CDX22}, with some adjustments to the details so that their paradigm fits our case. For sake of completeness, a full exposition is given in \cref{sec:min-plus-rect}.

\subsubsection*{Proving \cref{thm:main}.}
In \cref{sec:recursion}, we define a generic Valiant-like recursion (whose structure resembles that of the dynamic-programming algorithm based on limiting the options for $m$) and show that it can be simulated in time proportional (up to polylogarithmic factors) to a single $(\min,+)$ product.
As a straightforward corollary, our $\Oh(n^2k)$-time dynamic-programming algorithm
can be sped up to $\Ohtilde(n^{\bar{\omega}(\log_n k)})$ time.
In particular, this proves \cref{thm:main} for $k\ge \sqrt{n}$ (if $k\gg \sqrt{n}$, the runtime is actually better than claimed in \cref{thm:main}). In \cref{sec:main}, we apply the result of \cref{sec:recursion} to convert our $\Oh(n+k^5)$-time algorithm into an $\Ohtilde(n+k^{2\bar{\omega}(1/2)})$-time solution, thereby completing the proof of \cref{thm:main} for the case of $k<\sqrt{n}$.

\section{Combinatorial observations and $\Oh(n^2k)$-time algorithm}\label{sec:comb_algo}
In this section, we introduce a new combinatorial property (Lemma~\ref{lm:mid_point}) that allows the dynamic programming algorithm to spend only $O(k)$ time to compute each value of the matrix $\D$, and hence achieve $\Oh(n^2 k)$ runtime, which is an improvement over the runtime of Backurs and Onak's algorithm~\cite{BO16} for all $k = \Omega(n^{2/13})$.

\begin{definition}[Heights]\label{def:height}
	Define the function $h: \Sigma\to\{-1,1\}$ so that $h(a)=1$ if $a \in \Sigma$ is an opening parenthesis and $h(a)=-1$ otherwise. Given a sequence $S \in \Sigma^n$, define the height of a position $i$ of $S$, $0 \le i \le n$, as $H(i)=\sum_{j=0}^{i-1} h(S[j])$.
\end{definition}
Notice that $H(i)$ is the difference between the number of opening parentheses and the number of closing parentheses in $S[0\dd i)$.

\begin{definition}[Peaks and valleys]
Let $S\in\Sigma^n$. We say that a position $i\in [1\dd n)$ is a \emph{peak} if $H(i-1)<H(i)>H(i+1)$
and a \emph{valley} if $H(i-1)>H(i)<H(i+1)$.
\end{definition}

The following claim allows us to assume, without loss of generality, that there are at most $2k$ valleys.

\begin{claim}[{Corollary of~\cite[Claim 35]{BO16}}]\label{claim:number_of_valleys}
Let $S \in \Sigma^n$. There exists an algorithm that preprocesses $S$ in $\Oh(n)$ time, and either rejects $S$ (meaning that $\ED(S) > k$), or outputs a string $S'$ of length at most $n$ such that $\ED(S) = \ED(S')$ and $S'$ has at most $2k$ valleys.
\end{claim}

From now on, we assume that the input string $S$ has at most $2k$ valleys, and we denote the set of valleys by~$V$. This assumption leads us to refine \cref{eq:basic_recursion} by observing that instead of considering all possible values of $m$ for  $i<m<j$, it is enough to consider only values of  $m$  that are at distance at most 1 from some valley or at most 2 from $i$ or $j$.

\begin{lemma}\label{lm:mid_point}
	Let $M = \bigcup_{v \in V} \{v-1, v, v+1\}$ for $S\in \Sigma^n$. Let $i,j\in [0\dd n]$ with $j-i\ge 2$. Then
	\begin{equation}\label{eq:faster_recursion}
	\D[i,j] = \min
	\begin{cases}
	\D[i, m] + \D[m, j] \qquad \text{for } m\in (i\dd j)\cap (M\cup\{i+1,i+2,j-2,j-1\}),\\
	\ED(S[i]S[j-1]) + \D[i+1,j-1].
	\end{cases}
	\end{equation}
\end{lemma}
\newcommand{\mst}{p}
\begin{proof}
	Notice that
	if $\D[i,j] = \ED(S[i]S[j-1]) +\D[i+1,j-1]$ then, by the correctness of \cref{eq:basic_recursion}, the equality holds since the options for $m$ in \cref{eq:faster_recursion} are a subset of the options for $m$ in \cref{eq:basic_recursion}. Thus, we focus on the case where $\D[i,j] = \D[i,m]+\D[m,j]$ for some $m \in (i\dd j)$.
	Let $\mst =\arg \min \{H(m) : m\in (i\dd j)\text{ and } \D[i,j] = \D[i,m]+\D[m,j]\}$. That is, $\mst$ is an index from $(i\dd j)$ for which $\D[i,j] = \D[i,\mst ]+\D[\mst,j]$ and $H(\mst )$ is minimized.
	In the following, we complete the proof by showing that $\mst \in M \cup \{i+1,i+2,j-2,j-1\}$.
	For a proof by contradiction, suppose that $\mst \notin M \cup \{i+1,i+2,j-2,j-1\}$.
	Then, $[\mst -2\dd \mst +2]\sub (i\dd j)$ since $\mst \notin \{i+1,i+2,j-2,j-1\}$, and $\min\{H(\mst -2),H(\mst +2)\}=H(\mst )-2$ since $\mst \notin M$.
	By symmetry, we assume without loss of generality that $H(\mst +2)=H(\mst )-2$.
	In this case,  $H(\mst +1) = H(\mst )-1$, and $S[\mst ],S[\mst +1]$ are both closing parentheses.

	Let $q\in (\mst \dd j]$ be the smallest value such that $\D[\mst,j]=\D[\mst,q]+\D[q,j]$ (notice that $\D[\mst,j]=\D[\mst,j]+\D[j,j]$, so $q$ is always well-defined).
	We consider two cases:

	\textbf{Case 1: $q\le \mst +2$.}
	In this case, $\D[i,j] = \D[i,\mst ]+\D[\mst,q] + \D[q,j] \ge \D[i,q] + \D[q,j] \ge \D[i,j]$ and hence $\D[i,j] = \D[i,q] + \D[q,j]$.
	However, $H(q) < H(\mst )$, which contradicts the choice of $\mst $ since $H(\mst )$ is not minimized.

\textbf{Case 2: $q > \mst +2$.} Recall that $q$ is minimized. Therefore,  there is no $q'\in (p\dd q)$ such that $\D[\mst,q] = \D[\mst,q']+\D[q',q]$ as otherwise $\D[\mst,j]=\D[\mst,q]+\D[q,j] = \D[\mst,q']+\D[q',q] +\D[q,j] \ge D[\mst,q'] +\D[q',j] \ge \D[\mst,j]$ and hence $\D[\mst,j] = D[\mst,q'] +\D[q',j]$, a contradiction. Consequently, we have $\D[\mst,q] = \ED(S[\mst ]S[q-1]) + \D[\mst +1,q-1]$ and $\D[i,j] = \D[i,\mst ] + \ED(S[\mst ]S[q-1]) + \D[\mst +1,q-1] + \D[q,j]$.
Let us define $r\in ( \mst+1 \dd q-1]$ as the smallest value such that $\D[\mst +1,q-1]=\D[\mst +1,r]+\D[r,q-1]$.
We consider two subcases:

\textbf{Case 2a: $r = \mst +2$.}
Notice that $\ED(S[\mst ]S[{q-1}]) = \ED(S[\mst +1]S[{q-1}])$ since $S[\mst ]$ and $S[\mst +1]$ are closing parentheses.
Recall that for $\ell \in [0\dd n)$ we have $\D[\ell,\ell +1] =1$.
Consequently,
\begin{align*}\D[i,j] &=  \D[i,\mst ] + \ED(S[\mst ]S[q-1]) + \D[\mst +1, {\mst +2}]+\D[\mst +2,q-1] + \D[q,j] \\
	 &= \D[i,\mst ] + \D[\mst,\mst +1] + \ED(S[\mst +1]S[q-1]) +\D[\mst +2,q-1] + \D[q,j] \\
	 &\ge \D[i,\mst +1] + \D[\mst +1,q]  +\D[q,j] \\
	 &\ge  \D[i,\mst +1]+\D[\mst +1,j]\\
	 &\ge \D[i,j].\end{align*}
Thus, $\D[i,j]=\D[i,\mst +1]+\D[\mst +1,j]$.
However, $H(\mst +1) < H(\mst )$, which contradicts the choice of $\mst $ since $H(\mst )$ is not minimized.

\textbf{Case 2b: $r > \mst +2$.}
Notice that  $\ED(S[\mst ]S[\mst +1]) \le \ED(S[\mst +1]S[r-1])$ since $S[\mst ]$ and $S[\mst +1]$ are closing parentheses.
Moreover, if $S[q-1]$ is a closing parenthesis, then $\ED(S[r-1]S[q-1])\le 1 = \ED(S[\mst ]S[q-1])$. Otherwise, $S[q-1]$ is an opening parenthesis and $\ED(S[r-1]S[q-1])\le 2  = \ED(S[\mst ]S[q-1])$.
Notice that for $\ell \in [0\dd n-1)$ we have $\D[\ell,\ell+2] =\ED(S[\ell ]S[\ell +1])$.
Consequently,
\begin{align*}
	\D[i,j] &= \D[i,\mst ] + \ED(S[\mst ]S[q-1]) + \ED(S[\mst +1]S[r-1]) + \D[\mst +2,r-1]+\D[r,q-1] + \D[q,j] \\
	&\ge \D[i,\mst ] + \ED(S[\mst ]S[\mst +1]) + \D[\mst +2,r-1] + \ED(S[r-1]S[q-1]) +\D[r,q-1] + \D[q,j] \\
	& \ge \D[i,\mst ] + \D[\mst, \mst +2]  + \D[\mst +2,r-1] + \D[r-1,q] + \D[q,j] \\
	& \ge \D[i,\mst +2] + \D[\mst +2,j]\\
	& \ge \D[i,j].
\end{align*}
Thus, $\D[i,j]=\D[i,\mst +2]+\D[\mst +2,j]$.
However, $H(\mst +2) < H(\mst )$, which contradicts the choice of $\mst $ since $H(\mst )$ is not minimized.
\end{proof}

By~\cref{lm:mid_point} it is straightforward to prove the following corollary.
\begin{corollary}
There exists an algorithm for \TDED with $O(n^2k)$ running time.
\end{corollary}

\section{$\Oh(n+k^5)$-time algorithm}\label{sec:alg}
In this section, we further refine the dynamic programming approach to achieve an $\Oh(n+k^5)$-time algorithm. To do so, the algorithm carefully selects a subset of $O(k^4)$ entries of the matrix $\D$ based on~\cref{lm:mid_point}, and the remaining entries are computed using a much faster method based on the approach developed by Landau and Vishkin for computing the threshold edit distance (\cite{LV97}, see also ~\cite{LMS98}).

\subsection{Definitions and combinatorial observations.}
We use the following simple fact that is an extension of~\cite[Fact 36]{BO16}.

\begin{fact}\label{fact:height_diff}
For $0\le i \le j\le n$, if $\D[i,j] \le k$, then $H(m) \ge \max\{H(i), H(j)\} - 2k$ holds for every $m\in [i\dd j]$.
In particular, $|H(i)-H(j)| \le 2k$.
\end{fact}

\begin{definition}[Trapezoid]
For $0\le a < b \le c < d \le n$, we say that $(a,b,c,d)$ is a \emph{trapezoid} if the following properties hold:
(1) $H(a)=H(d)$,
(2) $b-a = H(b)-H(a)=H(c)-H(d)=d-c$ (this value is called the \emph{height} of the trapezoid), and
(3) for all $m\in [b\dd c]$, we have $H(m) \ge H(b)=H(c)$. See \cref{fig:trapezoids}.
\end{definition}
Notice that all the characters in $S[a\dd b)$ are opening parentheses
and all the characters in $S[c\dd d)$ are closing parentheses, respectively.

\begin{definition}[Tall and maximal trapezoid]\label{def:trapezoid}
A trapezoid $(a,b,c,d)$ is \emph{maximal} if it cannot be \emph{extended} to a different trapezoid $(a',b',c',d')$
with $a' \le a $, $b' \ge b$,  $c' \le c$, and $d'\ge d$.
A trapezoid is \emph{tall} if its height is at least $2k$.
\end{definition}

\begin{fact}\label{fct:disjoint}
If $(a,b,c,d)$ and $(a',b',c',d')$ are distinct maximal trapezoids,
then $[a\dd b)\cup [c\dd d)$ and  $[a'\dd b')\cup [c'\dd d')$ are disjoint.
\end{fact}
\begin{proof}
	Suppose that these sets contain a common position $i$.
	By symmetry, assume without loss of generality that $i\in [a\dd b)$,
	which means that $S[i]$ is an opening parenthesis and $i\in [a'\dd b')$.
	Let $j\in (c\dd d]$ with $i+j=a+d$,
	and so $H(i)=H(j)$.
	Moreover, since $(a,b,c,d)$ is a trapezoid, for $m\in (i\dd j)$ we have  $H(m)>H(i)$.
	Symmetrically, define $j'\in (c'\dd d']$ so that $i+j'=a'+d'$. Thus,  $H(i)=H(j')$ and  for $m\in (i\dd j')$ we have $H(m)>H(i)$.
	This proves that $j=j'$ and that $(a,b,c,d)$ and $(a',b',c',d')$ both extend $(i,i+1,j-1,j)$.
	Finally, we note that $a'\ge a$ (otherwise $(a-1,b,c,d+1)$ is a trapezoid),
	$a'\le a$ (otherwise $(a'-1,b',c',d'+1)$ is a trapezoid),
	$b'\le b$ (otherwise $(a,b+1,c+1,d)$ is a trapezoid),
	and $b'\ge b$ (otherwise $(a',b'+1,c'-1,d')$ is a trapezoid).
	Hence, $(a,b,c,d)=(a',b',c',d')$ holds as claimed.
\end{proof}

We process tall maximal trapezoids using a fast method based on~\cite{LV97,LMS98}. The remaining positions are grouped into clusters, and processed using~\cref{lm:mid_point}. Formally,

\begin{definition}[Cluster]\label{def:clusters}
Consider a cycle with vertices $[0\dd n]$ and edges $\{(i,i+1) : i\in [0\dd n)\}\cup\{(n,0)\}$. For each tall maximal trapezoid $(a,b,c,d)$, delete edges $(i,i+1)$ for $i\in [a\dd b)\cup [c\dd d)$,
delete vertices $(a\dd b)\cup (c\dd d)$, and add edges $(a,d)$ and $(b,c)$. The connected components of the resulting graph are called \emph{clusters}.
We say that a cluster $C$ contains all positions $S[p]$ such that $(p,p+1)$ is an edge of $C$. See~\cref{fig:trapezoids-and-clusters}.
\end{definition}

\begin{lemma}\label{lm:size_trapezoids_clusters}
If the input string $S$ has at most $2k$ valleys and $\D[0,n]\le k$, then the number of maximal trapezoids is $\Oh(k)$ and the total number of positions in the clusters is $\Oh(k^2)$.
\end{lemma}
\begin{proof}
Consider a maximal trapezoid $(i,i',j',j)$. First, note that $i\in V\cup\{0\}$ or $j\in V\cup\{n\}$ because $(i-1,i',j',j+1)$ is not a trapezoid. Moreover, by \cref{fct:disjoint}, there is no other maximal trapezoid $(a,b,c,d)$ with $a=i$ or $d=j$. Hence, there are at most two maximal trapezoids associated with each $v\in V$, and another maximal trapezoid associated with each extreme position (0 and $n-1$).
Thus, number of maximal trapezoids is bounded by $4k+2$.

We now show that all positions in clusters, except for $\Oh(k)$ of those positions, belong to some maximal trapezoid.
Consider an opening parenthesis $S[p]$. If there is a position $r\in [p\dd n]$ such that $H(r)=H(p)$,
then choose the leftmost such position. In this case, $S[r-1]$ is a closing parenthesis and $(p,p+1,r-1,r)$
forms a trapezoid, which can be extended to a maximal trapezoid.
Thus, an opening parenthesis $S[p]$ does not belong to any maximal trapezoid only if $H(p)<H(r)$ for all $r\in [p\dd n]$. The number of such positions does not exceed $H(n)-\min_{p\in [0\dd n]} H(p)$, since each height between $\min_{p\in [0\dd n]} H(p)$ and $H(n)$ can contribute at most once.
By \cref{fact:height_diff}, $H(n)-\min_{p\in [0\dd n]} H(p) \le 2k $. Thus, the number of such positions is at most $2k$.
Similarly, there are at most $2k$ closing parentheses that do not belong to any maximal trapezoid.

Finally, we note that if a maximal trapezoid is not tall, then it contains at most $4k$ positions.
Since the number of maximal trapezoids is $\Oh(k)$, this means that maximal trapezoid that are not tall contain at most $\Oh(k^2)$ positions in total. Overall, we conclude that the total number of positions in the clusters is $\Oh(k^2)$ as desired.
\end{proof}

\begin{definition}[Parent relation on clusters and trapezoids]
	Let $T=(a,b,c,d)$ be a tall maximal trapezoid and let $C$ be a cluster.
	We say that $C$ is the parent of $T$ if $a,d\in C$,
	and that $T$ is the parent of $C$ if $b,c\in C$. See~\cref{fig:the-tree}.
\end{definition}

	The following lemma shows that clusters and trapezoids naturally define a tree structure.
	
	\begin{lemma}\label{lm:tree}
	If $\D[0,n] \le k$, then the parent relation on clusters and tall maximal trapezoids defines a rooted tree whose root is the cluster containing both $0$ and $n$. We denote this tree by $\T$.
	
	Moreover, there exists an algorithm for constructing tree $\T$ in $\Oh(n)$ time.
	\end{lemma}
	\begin{proof}
	
		For a tall maximal trapezoid $T=(a,b,c,d)$, the graph $G$ of \cref{def:clusters}
		does not contain any edge connecting $[b\dd c]$ with $[0\dd a]\cup [d\dd n]$.
		Consequently, the cluster $C$ whose parent is $T$ satisfies $\min C = b$ and $\max C = c$.
		At the same time, the parent $C'$ of $T$ satisfies $\min C' \le a < b$ and $\max C' \ge d > c$.
		We conclude that the parent relation is acyclic.
		
		It remains to prove that every cluster $C$ (except for the one containing $0,n$) has exactly one parent.
		Consider an edge $(\max C, \max C + 1)$ \emph{not} present in $G$.
		Since $\max C\in G$, this edge must have been deleted while processing a tall maximal trapezoid $T=(a,b,c,d)$ such that $\max C\in \{a,c\}$.
		If $\max C = a$, then $(a,d)$ is an edge of $G$, so $d \in C$ and $\max C \le d < a = \max C$, a contradiction.
		Consequently, we must have $\max C = c$, which means that~$T$ is the parent of $C$.
		This trapezoid is unique due to \cref{fct:disjoint}.

	We now show an algorithm that builds $\T$ in $\Oh(n)$ time.
	First, for each opening parenthesis $S[p]$, we compute the value $B(p)=\min \{ r\in (p\dd n] : H(p)=H(r)\}$, assuming $\min \emptyset = \infty$.
	For this, it suffices to scan $S$ from left to right in linear time.
	Now, observe that $(a,b,c,d)$ is a trapezoid if and only if $i+B(i)=a+d=b+c$ for all $i\in [a\dd b)$.
	In particular, maximal trapezoids correspond to maximal segments $S[a\dd b)$ of opening parentheses such that $i+B(i)$ is fixed for all $i\in[a\dd b)$. This allows constructing all maximal trapezoids in $\Oh(n)$ time.
	Out of these trapezoids, we filter tall trapezoids (for which the height is at least $2k$).
	
	We now show how to compute the clusters.
	Let $G$ be a graph on nodes $[0\dd n]$ that contains edges $(i,i+1)$ for each $i\in [0\dd n)$ and an edge $(n,0)$ (i.e., $G$ is a cycle on $[0\dd n]$).
	Next, for each tall maximal trapezoid $(i,i',j',j)$, modify $G$ as specified in \cref{def:clusters}.
	By \cref{fct:disjoint}, no vertex or edge is deleted twice.
	Finally, traverse the graph and mark each vertex with an identifier of its connected component (cluster).
	Thus, constructing all the clusters costs $\Oh(n)$ time in total.
	
	It remains to build the tree $\T$. For this, each tall maximal trapezoid $T=(a,b,c,d)$ is connected to its child (the cluster containing $b$ and $c$) and its parent (the cluster containing $a$ and $d$). Overall, the algorithm costs $\Oh(n)$ time.
\end{proof}

Notice that every trapezoid in $\T$ has a cluster parent, and at most 1 child, which must be a cluster. Moreover, every cluster in $\T$ has an arbitrary number of trapezoid children, and at most 1 trapezoid parent (the root of $\T$ is a cluster and has no parent).

\begin{claim}\label{claim:trapezoids}
Consider a trapezoid $(a,a+2k,d-2k,d)$ and indices $i\le j$ such that $\D[i,j]\le k$.
If one of the indices~$i,j$ belongs to $(a+2k\dd d-2k)$, then the other belongs to $(a\dd d)$.
\end{claim}
\begin{proof}
By symmetry, we assume without loss of generality that $i\in (a+2k\dd d-2k)$.
Then, $H(i)>H(d-2k)=H(d)+2k$ and, by \cref{fact:height_diff}, $d\notin [i\dd j]$,
i.e., $j\in (a\dd d)$.
\end{proof}

\subsection{The Algorithm.}
Our main algorithm constructs $\T$ and then performs a post-order traversal of $\T$,
with distinct procedures for processing nodes representing clusters and trapezoids.

\subsubsection*{Processing clusters.}
When the algorithm processes a cluster $C$ with $r$ children $T_{q} = (a_{q},b_{q},c_{q},d_{q})$ for $q\in [1\dd r]$, the algorithm is given the values $\min\{\D[i,j], k+1\}$ for all $i,j\in [a_{q}\dd a_q+2k] \cup [d_q-2k\dd d_q]$ and $q\in [1\dd r]$. If $C$ has a parent $T = (a,b,c,d)$, then the outcome of processing $C$ consists of the values $\min\{\D[i,j], k+1\}$ for $i,j \in [b-2k\dd b] \cup [c\dd c+2k]$. See~\cref{fig:cluster-inputs-output}.
 If $C$ is the root of~$\T$, then $0,n\in C$ by~\cref{lm:tree} and the output is $\D[0,n]$.
Denote
\[E(C) :=  C \cup \left(\bigcup_{q=1}^r [a_{q}\dd a_q+2k] \cup [d_q-2k\dd d_q] \right) \cup \left([b-2k\dd b] \cup [c\dd c+2k]\right) = [b-2k\dd c+2k] \sm \bigcup_{q=1}^r  (a_q+2k \dd d_q-2k),\]
where the last term is included only if $C$ is not the root of $\T$.
Notice that $E(C)$ is the set of positions that are relevant for processing $C$. Let $M(C) = M \cap E(C)$, where $M$ is defined as in \cref{lm:mid_point}.
The following lemma is a simplification of \cref{lm:mid_point}, confined to values in $E(C)$ which are not given as input to cluster processing.
\begin{lemma}\label{clm:enough}
Let $(i,j)\in E(C)^2 \sm \bigcup_{q=1}^r [a_{q}\dd d_q]^2$.
If $j\ge i+2$ and $\D[i,j]\le k$, then
\begin{equation}\label{eq:terrible_recursion}
	\D[i,j] = \min\begin{cases}
		\D[i,m]+\D[m,j] \qquad \text{for }m\in (i\dd j)\cap (M(C)\cup\{i+1,i+2,j-2,j-1\}),\\
		\ED(S[i]S[j-1]) + \D[i+1,j-1] \qquad \text{if }i+1,j-1\in E(C).\\
	\end{cases}
\end{equation}
\end{lemma}
\begin{proof}
If $\D[i,j] = \D[i+1,j-1] + \ED(S[i]S[j-1])\le k$, then
observe that $i+1,j-1\in E(C)$ unless $i+1\in (a_q+2k\dd d_q-2k)$ or $j-1\in (a_q+2k\dd d_q-2k)$ for some $q\in [1\dd r]$. By \cref{claim:trapezoids} applied to $(a_q,a_q+2k,d_q-2k,d_q)$, the assumption $\D[i,j]\le k$ implies $\D[i+1,j-1]\le k$ and so $(i+1,j-1)\in (a_q\dd d_q)^2$, i.e., $(i,j)\in [a_q\dd d_q]^2$, a contradiction.

Thus, by \cref{lm:mid_point}, we focus on the case where $\D[i,j] = \D[i,m]+\D[m,j]$ for some $m \in (i\dd j)\cap (M\cup \{i+1,i+2,j-2,j-1\})$.
Assume by contradiction that $m\notin M(C)$, and so $m\in (a_q+2k\dd d_q-2k)$ for some $q\in [1\dd r]$.
By \cref{claim:trapezoids} applied to $(a_q,a_q+2k,d_q-2k,d_q)$, the assumption $\D[i,m]+\D[m,j]\le k$ implies $i,j\in (a_q\dd d_q)$.
However,  $(i,j)\notin [a_{q}\dd d_q]^2$, and so we obtain a contradiction. \end{proof}

By \cref{clm:enough}, the values $\min\{\D[i,j],k+1\}$ for all $i,j\in E(C)$ with $i\le j$
can be computed in $\Oh(|E(C)|^2 (1+|M(C)|))$ time.
Notice that $|E(C)|\le |C|+(4k+2)(q+1)$.
Moreover, by \cref{lm:size_trapezoids_clusters}, the sum of the number of children for all clusters is $\Oh(k)$ and the total number of positions in the clusters is $O(k^2)$, and therefore $\sum_{C} |E(C)| = \Oh(k^2)$.
Furthermore, $|M(C)\sm C| = \Oh(q+1)$ (because, for each trapezoid $(a,b,c,d)$, only $a,a+1,d-1,d$ may belong to $M$) and $\sum_C |M(C) \cap C| \le |M| = \Oh(k)$ as the clusters do not intersect,
so $\sum_C |M(C)| = \sum_C |M(C) \sm C| + \sum_C |M(C) \cap C| = \Oh(k)$. Consequently, the total time required to process all clusters is $\Oh((k^2)^2\cdot (1+k)) = \Oh(k^5)$.

\subsubsection*{Processing trapezoids.}
When we process a tall maximal trapezoid $T = (a,b,c,d)$, we are given values $\min\{\D[i,j], k+1\}$ for all $i,j \in [b-2k\dd b]\cup [c\dd c+2k]$, and the output consists of the values $\min\{\D[i,j], k+1\}$ for $i,j \in [a\dd a+2k]\cup [d-2k\dd d]$. See~\cref{fig:trapezoid-input-output}.

For $(i,j)\in [a\dd b]^2 \cup [c\dd d]^2$, we have $\D[i,j] = \lceil (j-i)/2 \rceil$
because all parentheses in $S[i\dd j)$ are of the same orientation. We therefore only need to compute entries $\D[i,j]$ such that $i \in [a\dd a+2k]$ and $j\in [d-2k\dd d]$. Consider a square submatrix $\D[a\dd b, c \dd d]$ of $\D$ (see~\cref{fig:submatrix}). We start by showing that the values in the submatrix
follow a simpler recursion:
\begin{claim}\label{claim:diag}
For all $(i,j) \in [a\dd d]^2 \sm [b-2k\dd c+2k]^2$ such that $i \le j$ and $\D[i,j] \le k$,  \[\D[i,j] = \min\{\D[i+2,j] + 1, \D[i+1,j] + 1, \D[i+1,j-1] + \ED(S[i]S[j-1]), \D[i,j+1]+1, \D[i,j+2]+1\}.\]
\end{claim}
\begin{proof}
Note that $\D[i,i+2] = \D[i,i+1] = \D[j-1,j]=\D[j-2,j]=1$, since for each computation all the parentheses that matter are of the same orientation. Hence, it suffices to prove that
if $\D[i,j]=\D[i,m] + \D[m,j]$ then, without loss of generality, $m\in \{i+1,i+2,j-2,j-1\}$, since, for example, if $m=i+1$ then $\D[i,j]=\D[i,i+1] + \D[i+1,j] = 1+\D[i+1,j]$.
By \cref{lm:mid_point}, the only other possibility is that $m\in M$
and, since $(a\dd d)\cap V \sub (b\dd c)$, this implies that $m\in [b\dd c]$.
However, $\D[i,m]\le k$ would then contradict \cref{claim:trapezoids} for a trapezoid $(b-2k-1,b-1,c+1,c+2k+1)$.
\end{proof}

\begin{figure}
\begin{center}
\begin{tikzpicture}[scale=0.5]
\draw (0,0) rectangle (10,10);
\draw[fill=green,fill opacity=0.5] (0,10) rectangle (4,6);

\draw[fill=green,fill opacity=0.5] (10,0) rectangle (6,4);

\draw[fill=red] (6,6) rectangle (7,7);

\draw[yellow,very thick,->] (6,7)--(9,10);
\draw[yellow,very thick,->] (6.5,7)--(9,9.5);
\draw[yellow,very thick,->] (7,7)--(9,9);
\draw[yellow,very thick,->] (7,6.5)--(9.5,9);
\draw[yellow,very thick,->] (7,6)--(10,9);

\draw (3,7) rectangle (7,3);
\draw[fill=blue] (9,9) rectangle (10,10);

\node[above,rotate=45] at (0.25,10) {$a$};
\node[above,rotate=45] at (3.75,10.55) {$b-2k$};
\node[above,rotate=45] at (4.25,10) {$b$};
\node[above,rotate=45] at (6.25,10) {$c$};
\node[above,rotate=45] at (7.5,10.55) {$c+2k$};
\node[above,rotate=45] at (9.75,10.55) {$d-2k$};
\node[above,rotate=45] at (10.25,10) {$d$};

\node[left] at (0,9) {$a+2k$};
\node[left] at (0,7) {$b-2k$};
\node[left] at (0,6) {$b$};
\node[left] at (0,4) {$c$};
\node[left] at (0,3) {$c+2k$};
\node[left] at (0,0.25) {$d$};
\end{tikzpicture}
\end{center}
\caption{The submatrix of $\D$ corresponding to a tall maximal trapezoid $T = (a,b,c,d)$. Green squares correspond to the values $\D[i,j]$ with $(i,j) \in [a\dd b] ^2 \cup [c\dd d]^2$, for which the values are known. Except for the green squares, only the values in the red square, yellow diagonals, and the blue square can be smaller or equal to $k$.  The processing of $T$ must compute the values in the blue square given the values in the red square.}
\label{fig:submatrix}
\end{figure}
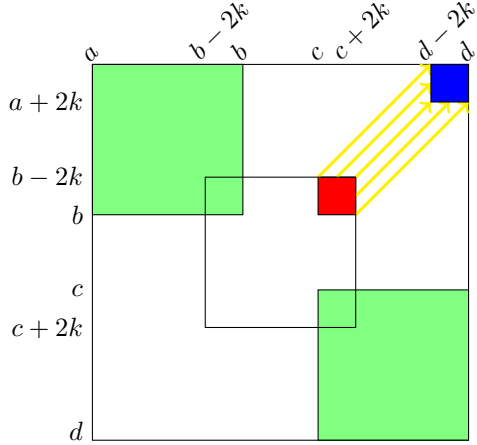

Consider an entry $\D[i,j] \le k$, where $(i,j) \in [a\dd d]^2 \sm  [b-2k\dd c+2k]^2$. We have either $(i,j) \in [a,b-2k)^2 \cup (c+2k,d]^2$ or, by~\cref{fact:height_diff}, $i+j = b+c+\delta$, where $\delta \in [-2k \dd 2k]$. Let us denote $\delta_+:=\max(0,\delta)$ and
$\delta_- := \min(0,\delta)$. For a fixed $\delta \in [-2k \dd 2k]$, we call the set of entries $(i,j)  \in [a\dd d]^2 \sm  [b-2k\dd c+2k]^2$ such that $i+j = b+c+\delta$ \emph{the diagonal $\delta$. } Let us show that the values $\D[i,j]$ are monotone along each diagonal:

\begin{claim}\label{claim:monotonicity}
	For all $(i,j) \in [a\dd d]^2 \sm [b-2k\dd c+2k]^2$ such that $\D[i,j] \le k$, we have $\D[i,j] \ge \D[i+1,j-1]$.
\end{claim}
\begin{proof}
	Let $q\in (i\dd j]$ be the smallest value such that $\D[i,j]=\D[i,q]+\D[q,j]$.
	We consider four cases:
	\begin{enumerate}
		\item $q=j$. In this case, $\D[i,j] = \ED(S[i]S[j-1]) + \D[i+1,j-1]\ge \D[i+1,j-1]$.
		\item $q=j-1$. In this case, $\D[i,j] = \ED(S[i]S[j-2]) + \D[i+1,j-2]+\D[j-1,j] \ge \D[i+1,\allowbreak {j-2}]+\D[j-2,j-1] \ge \D[i+1,j-1]$.
		\item $q = i+1$. Let $r\in [q\dd j)$ be the largest value such that $\D[q,j]=\D[q,r]+\D[r,j]$.
		We consider two subcases:
		\begin{enumerate}
			\item $r = j-1$. In this case, $\D[i,j] = \D[i,i+1] + \D[i+1,j-1] + \D[j-1,j] \ge \D[i+1,j-1]$.
			\item $r < j-1$. In this case, $\D[i,j] = \D[i,i+1] + \D[i+1,r] + \ED(S[r]S[j-1]) + \D[r+1,j-1]
			\ge \D[i+1,r] + \D[r,r+1] + \D[r+1,j-1] \ge \D[i+1,j-1]$.
		\end{enumerate}
		\item $i+1 < q < j-1$. Again, let $r\in [q\dd j)$ be the largest value such that $\D[q,j]=\D[q,r]+\D[r,j]$.
		We consider two subcases:
		\begin{enumerate}
			\item $r = j-1$. In this case,  $\D[i,j] = \ED(S[i]S[q-1]) + \D[i+1,q-1] + \D[q,j-1] +\D[j-1,j]
			\ge \D[i+1,q-1]+ \D[q-1,q] + \D[q,j-1]\ge \D[i+1,j-1]$.
			\item $r < j-1$. In this case,  $\D[i,j] = \ED(S[i]S[q-1]) + \D[i+1,q-1] + \D[q,r] + \ED(S[r]S[j-1])
			+ \D[r+1,j-1] \ge \D[i+1,q-1] + \ED(S[q-1]S[r]) + \D[q,r] + \D[r+1,j-1] \ge \D[i+1,q-1]+\D[q-1,r+1]+\D[r+1,j-1]\ge \D[i+1,j-1]$, where $\ED(S[i]S[q-1]) + \ED(S[r]S[j-1]) \ge 1 \ge \ED(S[q-1]S[r])$ follows from the fact that $S[q-1]$ is an opening parenthesis or $S[r]$ is a closing parenthesis (otherwise, we would have $q,r\in (b\dd c)$, which contradicts $\D[i,q]\le k$ and $\D[r,j]\le k$ by \cref{claim:trapezoids} applied to $(b-2k,b,c,c+2k)$).
		\end{enumerate}
	\end{enumerate}
\end{proof}

For the diagonal $\delta$ and values $v \in [0\dd k]$, we shall compute values $L_v[\delta] =  \max\{ j\in [c+2k+\delta_-\dd d-2k+\delta_+] : \D[b+c+\delta-j,j] \le  v\}$, assuming $\max \emptyset = -\infty$. From the resulting tables $L_v[\delta]$, we can determine $\min(\D[i,j],k+1)$ for all $i,j \in [a\dd a+2k] \cup [d-2k\dd d]$ in $\Oh(k^2)$ time by~\cref{claim:diag}. To compute the values $L_v[\delta]$, we rely on the fact that, after $\Oh(n)$-time preprocessing, the longest common prefix of any two substrings of $S\overline{S}$ can be computed in $\Oh(1)$-time~\cite{10.1007/11780441_5}.\footnote{\vphantom{$2^{2^2}$}Here, $\overline{S}$ denotes the reverse complement of $S$, obtained by reversing $S$ and flipping the orientation of every parenthesis.}
Specifically, we use an adaptation of the Landau--Vishkin method~\cite{LV97,LMS98} generalizing a similar subroutine present in~\cite{BO16}. Our procedure is implemented as~\cref{alg:LV}, which assumes that $L_v[s]=-\infty$ for all out-of-bounds and uninitialized entries.

\begin{algorithm}\label{alg:LV}
	\begin{algorithmic}[1]\caption{Processing trapezoids.}
		\FOR{$v := 0$ to $k$}
		\FOR{$\delta := -2v$ to $2v$}
		\STATE $L'_v[\delta] := \min(d-2k+\delta_+,\max(L_{v-1}[\delta-2]+2, L_{v-1}[\delta-1]+2,L_{v-1}[\delta]+1, L_{v-1}[\delta+1], L_{v-1}[\delta+2])$
		\IF{$\D[b-2k+\delta_+,c+2k+\delta_-]\le v$}
		\STATE $L'_v[\delta]:= \max(L'_v[\delta], c+2k+\delta_-)$ 
		\ENDIF
		\STATE $L_v[\delta] := L'_v[\delta]+ \lcp(\overline{S[a\dd b+c+\delta - L'_v[\delta])},S[L'_v[\delta]\dd d))$
		\ENDFOR
		\ENDFOR
	\end{algorithmic}
\end{algorithm}

\begin{lemma}\label{lm:fix_BO}
	For all $(i,j) \in ([a \dd b]\times [c\dd d])\setminus ([b-2k\dd b]\times [c\dd c+2k])$ and $v\in [0\dd k]$, we have $\D[i,j]\le v$ if and only if the value $L_v[(i+j)-(b+c)]$ computed by \cref{alg:LV} is at least $j$.
\end{lemma}
\begin{proof}
	First, let us prove that $\D[i,j]\le v$ implies $L_v[\delta] \ge j$, where $\delta = (i+j)-(b+c)$.
	We proceed by induction on $v$. Let $S[i\dd i')=\overline{S[j'\dd j)}$ be the longest common prefix of $S[i\dd b-2k+\delta_+)$ and $\overline{S[c+2k+\delta_-\dd j)}$.
	We shall prove $L'_v[\delta] \ge j'$, where $\delta = (i+j)-(b+c)$, by considering six cases:
	\begin{itemize}
		\item If $j' = c+2k+\delta_-$ and $i'=b-2k+\delta_+$, then $\D[i',j']\le v$ by \cref{claim:monotonicity}. Hence, $L'_{v}[\delta] \ge c+2k+\delta_- = j'$.
	\end{itemize}
Otherwise, by~\cref{claim:diag}, it suffices to consider one of the remaining five cases:	
	\begin{itemize}
		\item If $\D[i',j']=\D[i',j'-2]+1$, then the inductive assumption yields $L'_v[\delta] \ge L_{v-1}[\delta-2]+2 \ge (j'-2)+2 = j'$.
		\item If $\D[i',j']=\D[i',j'-1]+1$, then the inductive assumption yields $L'_v[\delta] \ge L_{v-1}[\delta-1]+1 \ge (j'-1)+1 = j'$.
		\item If $\D[i',j']=\D[i'+1,j'-1]+1$, then the inductive assumption yields $L'_v[\delta] \ge L_{v-1}[\delta]+1 \ge (j'-1)+1 = j'$.
		\item If $\D[i',j']=\D[i'+1,j']+1$, then the inductive assumption yields $L'_v[\delta] \ge L_{v-1}[\delta+1] \ge j'$.
		\item If $\D[i',j']=\D[i'+2,j']+1$, then the inductive assumption yields $L'_v[\delta] \ge L_{v-1}[\delta+2] \ge j'$.
	\end{itemize}
	In all cases $L_v[\delta]\ge j$ follows from $L'_v[\delta]\ge j'$ due to $\lcp(\overline{S[a\dd i')},S[j'\dd d))\ge j-j'$.
	
	The converse implication is also proved by induction on $v$.
	Let $\delta = (i+j)-(b+c)$, $j' = L'_v[\delta]$ and $i' = i+j-j'$.
	By \cref{claim:monotonicity}, we can assume $j=L_v[\delta]$ without loss of generality.
	Moreover, since $\overline{S[i\dd i')}=S[j'\dd j)$, we have $\D[i,j]\le \D[i',j']$.
	Hence, it suffices to prove  $\D[i',j']\le v$. For this, we consider six cases:
	\begin{itemize}
		\item If $j'=c+2k+\delta_-$, then the algorithm explicitly checked $\D[i',j']\le v$ while setting $j'$.
		\item If $j' = L_{v-1}[\delta-2]+2$, then $\D[i',j'] \le \D[i',j'-2]+1\le (v-1)+1=v$ by the inductive assumption.
		\item If $j' = L_{v-1}[\delta-1]+1$, then $\D[i',j'] \le \D[i',j'-1]+1\le (v-1)+1=v$ by the inductive assumption.
		\item If $j' = L_{v-1}[\delta]+1$, then $\D[i',j'] \le \D[i'-1,j'-1]+1\le (v-1)+1=v$ by the inductive assumption.
		\item If $j' = L_{v-1}[\delta+1]$, then $\D[i',j'] \le \D[i'+1,j']+1\le (v-1)+1=v$ by the inductive assumption.
		\item If $j' = L_{v-1}[\delta+2]$, then $\D[i',j'] \le \D[i'+2,j']+1\le (v-1)+1=v$ by the inductive assumption.
	\end{itemize}
\end{proof}

\cref{lm:fix_BO} shows correctness of \cref{alg:LV}.  From the description of \cref{alg:LV}  it follows that the tables $L_v[\delta]$ and hence the entries $\min(\D[i,j],k+1)$ for all $i,j \in [a\dd a+2k] \cup [d-2k\dd d]$ can be computed in $\Oh(k^2)$ time.

\begin{theorem}\label{th:k_to_five}
Given a sequence $S$ of parentheses of length $n$, $\min\{\ED(S), k+1\}$ can be computed in $\Oh(n+k^5)$ time.
\end{theorem}
\begin{proof}
We start with an $\Oh(n)$-time preprocessing of $S$ according to~\cref{claim:number_of_valleys}. After this preprocessing, we can assume that $S$ contains at most $2k$ valleys. Next, we apply~\cref{lm:tree} to build a tree of clusters and trapezoids in $\Oh(n)$ time, and then we start the main phase of the algorithm. The total time used to process the clusters is $\Oh(k^5)$. The total time used to process the tall maximal trapezoids is $\Oh(n+k \cdot k^2) = \Oh(n+k^3)$. The theorem follows.
\end{proof}

\section{Valiant-like Recursion}\label{sec:recursion}
\newcommand{\complete}{\mathtt{complete}}
\newcommand{\compute}{\mathtt{compute}}
\newcommand{\update}{\mathtt{update}}

In this section, we describe how to speed up a dynamic-programming procedure that constructs
a matrix $A[0\dd n,0\dd n]$ so that, for some non-empty $M\sub [0\dd n]$:
\begin{enumerate}[label=(\arabic*)]
	\item \label{it:constant-comp} each entry $A[i,j]$ can be computed in $\Oh(1)$ time\label{it:dp}
	given ($\Oh(1)$-time access to) the remaining entries $A[i',j']$ with $i \le i' \le j' \le j$
	and the value $\min\{A[i,m]+A[m,j] : m \in (i\dd j)\cap M\}$,
	\item \label{it:BD} for every $m\in M$, subsequent entries $A[i,m]$ and subsequent entries $A[m,j]$ differ by at most~$1$.\label{it:bd}
\end{enumerate}
Our procedure is an adaptation of Valiant's approach~\cite{Valiant75} to
generalize computing a transitive closure by allowing a non-transitive operation.
The description of our adaptation follows a more recent perspective of Valiant's recursion
scheme as described by Okhotin~\cite{Okhotin14}.

The scheme recursively halves the domain $[0\dd n]$,
which yields a hierarchy of decompositions.
We call intervals comprising the $d$th decomposition \emph{level-$d$ intervals}
so that $[0\dd n]$ is the unique level-$0$ interval and
a level-$d$ interval is partitioned into two disjoint level-$(d+1)$ intervals.

The computation of $A[0\dd n,0\dd n]$ is based on two recursive procedures
that utilize an auxiliary array $P[0\dd n,0\dd n]$.
Their definitions below are adapted to our setting.
\begin{itemize}[leftmargin=*]
	\item[] $\compute(I)$: given a level-$d$ interval $I$, compute the entries $A[i,j]$ for $i,j\in I$ with $i\le j$.
	\item[] $\complete(I,J)$: given two level-$d$ intervals $I,J$ (with $I$ to the left of $J$),
	compute the entries $A[i,j]$ for $i\in I$ and $j\in J$, assuming that:
	\begin{itemize}
		\item the remaining entries $A[i,j]$ with $\min I\le i \le j \le \max J$ have already been computed,
		\item $P$ currently stores values $P[i,j]=\min\{A[i,m]+A[m,j]: m\in M \cap (\max I\dd  \min J)\}$ for each $i\in I$ and $j\in J$.
	\end{itemize}
\end{itemize}
The following procedure is the workhorse of the algorithm, allowing
speed-ups through fast matrix multiplication.
\begin{itemize}[leftmargin=*]
	\item[] $\update(I,K,J)$: given three distinct level-$d$ intervals $I,K,J$, set $P[i,j] := \min\{P[i,j],\min\{A[i,m]+A[m,j]: m\in K\cap M)\}\}$ for $i\in I$ and $j\in J$, assuming that all entries $A[i,m]$ and $A[m,j]$ with
	$i\in I$, $j\in J$, and $m\in K$ have already been computed.
\end{itemize}

In the original hierarchical decomposition, each level-$d$ interval of length at least $2$ is partitioned into two equal halves, whereas intervals of length $1$ form the recursion base case.
In this work, our goal is to make the algorithm more efficient when $|M|$ is much smaller than $n$,
and we need to control both $|I|$ and $|I\cap M|$ for the intervals $I$ in the decompositions.
Moreover, we do not need to partition intervals $I$ with $|I\cap M|=\Oh(1)$.
To facilitate these goals, we define the \emph{weight} of an interval $I\sub [0\dd n]$ as \[w(I):=|I\cap M|+\frac{|I\sm M|\cdot |M|}{n}\]
Note that the weight of $I$ is the sum of the singleton intervals constituting it. We also denote $W = w([0\dd n])$; note that $|M| \le W \le 2|M|$.
For $0\le d \le \floor{\log_2 W}$, we define a level-$d$
decomposition of $[0\dd n]$ into disjoint intervals.
The only level-$0$ interval is $[0\dd n]$.
For $0 \le d < \floor{\log_2 W}$, each level-$d$ interval $I$
is decomposed into two level-$(d+1)$ intervals $I_{\colindex}$ and $I_r$ (the left and right subinterval)
so that $w(I_\colindex)$ and $w(I_r)$ are as balanced as possible.
Since the weight of a singleton does not exceed $1$,
we can achieve $w(I_\colindex),w(I_r)\le \frac12(w(I)+1)$.
A simple induction shows that, for every level-$d$ interval $I$, we have
$w(I)\le 2^{-d}(W-1)+1$.
Consequently, $|I\cap M| \le w(I)\le 2^{1-d}W$,  $|I\cap M| \le w(I)\le 2^{2-d}|M|$, and $|I| \le \frac{n}{|M|}w(I) \le 2^{2-d}n$.

\begin{lemma}
The $\update$ function for level-$d$ intervals, $2 \le d\le \floor{\log_2 W}$, costs $\Oh((2^{-d} n)^{\bar{\omega}})$ time, where $\bar{\omega}$ is the exponent of \cref{thm:rectangular-min-plus} for $\alpha = \log_n |M|$.
\end{lemma}
\begin{proof}
This operation can be implemented via a min-plus product of a $|I|\times |K\cap M|$ matrix,
storing $A[i,m]$ for $i\in I$ and $m\in K\cap M$,
with a $|K\cap M|\times |J|$ matrix, storing $A[m,j]$ for $m\in K\cap M$ and $j\in J$.
By Property~\ref{it:BD}, the first matrix is $1$-column-BD, and the second matrix is $1$-row-BD, so
\cref{thm:rectangular-min-plus} can be applied to compute the product.
Note that $|I|,|J| \le 2^{2-d}n$ and $|K\cap M| \le 2^{2-d}|M| \le (2^{2-d}n)^{\log_n |M|}$.
Consequently, the product can be computed in $\Oh((2^{2-d} n)^{\bar{\omega}})$ time.
\end{proof}

\begin{lemma}
	The $\complete$ function for level-$d$ intervals, $1\le d \le \floor{\log_2 W}$, costs $\Oh((2^{-d}n)^{\bar{\omega}})$ time, where $\bar{\omega}$ is the exponent of \cref{thm:rectangular-min-plus} for $\alpha = \log_n |M|$.
\end{lemma}
\begin{proof}
If $d=\floor{\log_2 W}$, then the naive implementation costs $\Oh(|I|\cdot  |J|\cdot |M\cap (I \cup J)|)=\Oh((2^{-d}n)^2)$ time because $|M\cap (I \cup J)|=\Oh(1)$.
Otherwise, we proceed recursively based on the decompositions $I=I_{\colindex}\cup I_r$,
$J = J_{\colindex}\cup J_r$ into level-$(d+1)$ intervals:
\begin{enumerate}
	\item $\complete(I_r,J_\colindex)$,
	\item $\update(I_\colindex,I_r,J_\colindex)$,
	\item $\complete(I_\colindex,J_\colindex)$,
	\item $\update(I_r,J_\colindex,J_r)$,
	\item $\complete(I_r,J_r)$,
	\item $\update(I_\colindex,I_r,J_r)$,
	\item $\update(I_\colindex,J_\colindex,J_r)$,
	\item $\complete(I_\colindex,J_r)$.
\end{enumerate}
This sequence of steps is at the heart of Valiant's recursion~\cite{Valiant75,Okhotin14},
and it is easy to check that the pre-conditions for each application of  $\complete$
and $\update$ are satisfied. Since $I\times J = (I_r\times J_\colindex)\cup(I_r\times J_r)\cup (I_\colindex\times J_\colindex)\cup(I_\colindex\times J_r)$, the post-conditions of the four applications of $\complete$ guarantees
that the resulting entries $A[i,j]$ for $i\in I$ and $j\in J$ are determined correctly.
As for the running time, it suffices to observe that there are four applications of $\update$
and four recursive applications of $\complete$, all on level-$(d+1)$ intervals.
Since $\bar{\omega}>2$, the runtime is $\Oh((2^{-d}n)^{\bar{\omega}})$.
\end{proof}

\begin{lemma}\label{lem:compute}
	The $\compute$ function for a level-$d$ interval, $0\le d \le \floor{\log_2 W}$, costs $\Oh((2^{-d}n)^{\bar{\omega}})$ time, where $\bar{\omega}$ is the exponent of \cref{thm:rectangular-min-plus} for $\alpha = \log_n |M|$.
\end{lemma}
\begin{proof}
	If $d=\floor{\log_2 W}$, then the naive implementation costs $\Oh(|I|\cdot |I|\cdot |M\cap I|)=\Oh((2^{-d}n)^2)$ time because $|M\cap I|=\Oh(1)$.
	Otherwise, we proceed recursively based on the decomposition $I=I_{\colindex}\cup I_r$ into
	level-$(d+1)$ intervals:
	\begin{enumerate}
		\item $\compute(I_\colindex)$
		\item $\compute(I_r)$
		\item $\complete(I_\colindex,I_r)$
	\end{enumerate}
	Again, it is easy to see that the pre-condition for $\complete$ is satisfied.
	Since each $(i,j)\in I^2$ with $i\le j$ belongs to $I_{\colindex}^2\cup I_r^2\cup (I_{\colindex}\times I_r)$,
	the post-conditions of the two applications of $\compute$ and the one application of $\complete$
	guarantee that all entries $A[i,j]$ with $(i,j)\in I^2$ and $i\le j$ are computed correctly.
	As for the running time, it suffices to observe that there is one application of $\update$
and two recursive applications of $\complete$, all on level-$(d+1)$ intervals.
Since $\bar{\omega}>1$, the runtime is $\Oh((2^{-d}n)^{\bar{\omega}})$.
\end{proof}

In particular, \cref{lem:compute} applied to the initial recursive call $\compute([0\dd n])$
yields the following.%
\begin{corollary}\label{cor:valiant}
The values $A[i,j]$ for $0\le i \le j \le n$ can be computed in $\Ohtilde(n^{\bar{\omega}})$ time, where $\bar{\omega}$ is the exponent of \cref{thm:rectangular-min-plus} for $\alpha = \log_n |M|$.
\end{corollary}

\section{Faster Procedure for Clusters}\label{sec:main}
In this section, we apply \cref{cor:valiant} to develop a faster implementation of the subroutine for processing clusters. For this, we shall prove that the dynamic-programming procedure following from \cref{clm:enough}
satisfies the two conditions of Valiant recursion.
Formally, for a fixed cluster $C$, this procedure operates on an array $A[E(C),E(C)]$
resulting in $A[i,j]=\min\{\D[i,j],k+1\}$ for all $i,j\in E(C)$ with $j\ge i$.
Inspecting the recursion in \cref{clm:enough}, it is easy to see that this dynamic program satisfies condition~\ref{it:dp} of Valiant recursion.
As for condition~\ref{it:bd}, we need to prove that, for every $m\in M(C)$, the subsequent entries $A[i,m]$ (with $i\in E(C)$ and $i\le m$) and $A[m,j]$ (with $j\ge m$ and $j\in E(C)$) differ by at most one.
By symmetry, we focus without loss of generality on the former case.
Let $A[i,m]$ and $A[i',m]$ be such subsequent entries with $i<i'\le m$.
If $i'=i+1$, then $|\D[i,m]-\D[i+1,m]|\le 1$ because $S[i\dd m)$ and $S[i+1\dd m)$ are at distance~$1$.
Otherwise, we must have $i=a_q+2k$ and $i'=d_q-2k$ for some $q\in [1\dd r]$.
If $\D[i,m] < k$ or $\D[i',m]< k$, then \cref{claim:trapezoids} applied for a trapezoid $(a_q+1,a_q+2k-1,d_q-2k-1,d_q-1)$
implies $m\in (a_q+1\dd d_q-1)$, which contradicts $m\in M(C)$.
Consequently, $A[i,m],A[i',m]\in \{k,k+1\}$.
In both case $|A[i,m]-A[i',m]|\le 1$ holds as claimed.

Thus, \cref{cor:valiant} implies that $C$ can be processed in time $\max\{|E(C)|,|M(C)|^2\}^{\bar{\omega}}$,
where $\bar{\omega}$ is the exponent of \cref{thm:rectangular-min-plus} for $\alpha = \frac12$.
Due to $\sum_{C} |E(C)|=\Oh(k^2)$ and $\sum_{C} |M(C)|=\Oh(k)$,
the total cost of processing all clusters is $\Ohtilde(k^{2\bar{\omega}})$,
which is $\Oh(k^{4.544184})$ with high probability or $\Oh(k^{4.853059})$ deterministically.
The remaining components of the algorithm from \cref{sec:alg} cost $\Oh(n+k^3)$ time,
so this yields \cref{thm:main}.

\section{\boldmath$(\min,+)$-product of Rectangular Matrices with Partially Bounded Difference}\label{sec:min-plus-rect}
In this section we consider the problem of computing $(\min,+)$ product of rectangular matrices which are partially bounded difference.
We first describe a deterministic algorithm that follows the ideas of Bringmman et al.~\cite{BGSW19}.
We first describe a randomized algorithm follows the ideas of \cite{BGSW19} and then we describe how to derandomize this algorithm to get a deterministic algorithm.
Then, in \cref{sec:random_minplus} we describe a randomized algorithm which based on the recent result of Chi et al.~\cite{CDXZ22}, which improve upon the results of Bringmann et al.~\cite{BGSW19}, but does not introduce any derandomization method.

Bringmann et al.~\cite{BGSW19} designed algorithms for $(\min,+)$-product of two square matrices $A$ and $B$ of size $n\times n$ for two special cases.
In the first case, both input matrices are fully-BD matrices, and the runtime for this case is $\Ohtilde(n^{2.8244})$ randomized or $\Ohtilde(n^{2.8603})$ deterministic.
In the second case, only one of the input matrices is assumed to be either column-BD or row-BD, and the randomized runtime for the second case is $\Ohtilde(n^{2.9217})$ time.

In our setting, the two matrices are rectangular, the matrix $A$ is column-BD, and the matrix $B$ is row-BD.\@
One could extend the solution of Bringmann et al.~\cite{BGSW19} for the second case to work on rectangular matrices.
However, in our case, we take advantage of the additional structure of $B$ in order to further reduce the runtime.
Our algorithm is very similar to the algorithms of Bringmann et al.~\cite{BGSW19}, and we follow their structure.
We emphasize that our exposition is given here only for the sake of completeness; all the main ideas come from~\cite{BGSW19}.

The (first) algorithm of Bringmann et al.~\cite{BGSW19} considers the matrices $A,B$ and $C=A\star B$ as composed of blocks of size $\Delta\times\Delta$, for some parameter $\Delta$ which is a small polynomial in $n$, and $n$ is an integer multiple of $\Delta$.
The algorithm considers the right-bottom corner of each block as the \emph{representative of the block}.
The intuition is that, due to the bounded difference, the value of any two positions inside the same block differs by $O(\Delta)$, and so in order to compute a rough estimation of the output it suffices to consider only the representatives of the blocks.

Due to the following lemma, in our case where $A$ is column-BD and $B$ is row-BD, the matrix $C=A\star B$ is fully-BD.

\begin{restatable}{lemma}{obsCisbounded}\label{obs:c-bd}
	Let  $\alpha>0$.
	Let $A$ and $B$ be two integer matrices of sizes $n\times n^\alpha$  and $n^\alpha\times n$, respectively, such that $A$ is column-BD and $B$ is row-BD.\@
	Then, $C:=A\star B$ is fully-BD.	
\end{restatable}
\begin{proof}
	Let $(i,j)\in [n-1]\times[n]$. By definition, $C[i,j]=\min_{\colindex\in[n^\alpha]}\{A[i,\colindex]+B[\colindex,j]\}$.
	By the assumption on $A$, for any $(i,j)\in [n-1]\times[n^\alpha]$ we have $A[i,j]-1 \le A[i+1,j] \le A[i,j]+1$.
	Hence,
	\begin{align*}
	C[i+1,j]&=\min_{\colindex\in[n^\alpha]}\{A[i+1,\colindex]+B[\colindex,j]\}
	\le\min_{\colindex\in[n^\alpha]}\{(A[i,\colindex]+1)+B[\colindex,j]\}
	\\&=\min_{\colindex\in[n^\alpha]}\{A[i,\colindex]+B[\colindex,j]\}+1 =C[i,j]+1.
	\end{align*}
	Similarly,
	\begin{align*}
	C[i+1,j]&=\min_{\colindex\in[n^\alpha]}\{A[i+1,\colindex]+B[\colindex,j]\}
	\ge\min_{\colindex\in[n^\alpha]}\{(A[i,\colindex]-1)+B[\colindex,j]\}
	\\&=\min_{\colindex\in[n^\alpha]}\{A[i,\colindex]+B[\colindex,j]\}-1 =C[i,j]-1.
	\end{align*}
	Thus, $|C[i,j]-C[i+1,j]|\le 1$.
	The proof that $|C[i,j]-C[i,j+1]|\le 1$ is symmetric (based on the assumption on~$B$).
\end{proof}

The algorithm of Bringmann et al.~\cite{BGSW19} is composed of three phases. We first describe the randomized algorithm, and in \cref{sec:derandomize-min-plus} we describe how to derandomize the algorithm.

\subsection{Phase 1: Finding {\boldmath$\tilde C$} -- an approximation of {\boldmath$C$} with {\boldmath$O(\Delta)$} additive error.}
Let $\Delta$ be a positive integer parameter (to be fixed later) and assume that $n$ is divisible by $\Delta$ (we can always assume without loss of generality that both $n$ and $\Delta$ are powers of 2).
We partition $[n]$ into intervals of length $\Delta$, that are defined as follows:
for each $1\le i'\le n$ which is divisible by $\Delta$, let $I(i')=\{i\in[n]\mid i'-\Delta<i\le i'\}$.
For every $i',j'$ which are divisible by $\Delta$, the algorithm first computes $\tilde C[i',j']=C[i',j']$ exactly (naively) and then, for any $(i,j)\in I(i')\times I(j')$ the algorithm sets $\tilde C[i,j]\leftarrow C[i',j']$.
We emphasize that the definition of $\tilde C[i',j']$ in our algorithm differs from the definition in~\cite{BGSW19},  since in our case $A$ is not guaranteed to be row-BD and  $B$ is not guaranteed to be column-BD.\@
Moreover,  for pairs $i',j'$ which are divisible by $\Delta$ we define $\tilde C[i',j']$ to be exactly $C[i',j']$, while in~\cite{BGSW19}, $\tilde C[i',j']$ is an approximation of $C[i',j']$.
Nevertheless, the following lemma is similar in flavor to~\cite[{Lemma 2.1}]{BGSW19}.
\begin{lemma}\label{lem:approx-tilde-C}
	For $i',j'$ which are divisible by $\Delta$ and $(i,j)\in I(i')\times I(j')$ we have
	$|\tilde C[i,j]-C[i,j]|=|C[i',j']-C[i,j]|\le 2\Delta$.
\end{lemma}
\begin{proof}
	Notice that $\tilde C[i,j]=C[i',j']$ by definition, so we only need to prove that $|C[i',j']-C[i,j]|\le 2\Delta$.
	
	Notice that one can move from $C[i,j]$ to $C[i',j]$ in $i'-i< \Delta$ vertical steps.
	Moreover, one can move further from $C[i',j]$ to $C[i',j']$ in $j'-j< \Delta$ horizontal steps.
	Thus, one can move from $C[i,j]$ to $C[i',j']$ in less than $2\Delta$ steps.
	By \cref{obs:c-bd} each step increases the difference between $C[i,j]$ and the current entry by at most 1, hence the claim follows.
\end{proof}
\subsubsection*{Runtime.}
Notice that computing $\tilde C[i',j']$ for all $i',j'$ which are divisible by $\Delta$ costs $\Ohtilde((\frac n\Delta)^2n^\alpha)=\Ohtilde(\frac {n^{2+\alpha}}{\Delta^2})$ time.
In addition, the process of filling the rest of the matrix $\tilde C$ costs $O(n^2)$ time. Hence, the total running time of the first phase is $\Ohtilde(\frac {n^{2+\alpha}}{\Delta^2}+n^2)$ time.

\subsection{Phase 2: Reduction to $(\min,+)$-product with small entries.}

The following lemma is useful for the second phase.
\begin{lemma}[{cf.~\cite[{Lemma 1}]{AGM97}}]\label{lem:bounded-matrix-multiplication}
	Let $R\in\mathbb N$.
	Let $A$ and $B$ be $n\times s$ and $s\times n$ matrices, respectively, with entries in $\{-R,-R+1,\dots,R\}\cup\{\infty\}$.
	Then $A\star B$ can be computed in  $\Ohtilde(R\cdot M(n,s,n))$ time.
\end{lemma}

The second phase has $\rho$ iterations ($\rho$ will be fixed later as a small polynomial in $n$). In each iteration the algorithm chooses $i^r$ and $j^r$ independently and uniformly at random from $[n]$.
Let $A^r$ be the matrix where $A^r[i,\colindex]=A[i,\colindex]+B[\colindex,j^r]- \tilde C[i,j^r]$ and let $B^r$ be the matrix where $B^r[\colindex,j]=B[\colindex,j]-B[\colindex,j^r]+\tilde C[i^r,j^r]-\tilde C[i^r,j]$.
Let $C^r=A^r\star B^r$.
For each $i,j$, we have $C[i,j]=C^r[i,j]+\tilde C[i,j^r]-\tilde C[i^r,j^r]+\tilde C[i^r,j]$, and therefore one can compute $C$ from $C^r$.
However, computing $C^r$ exactly seems to be inefficient, and so~\cite{BGSW19} proved that a partial computation of $C^r$ can be done efficiently enough and still result in useful information regarding $C$.

\subsubsection*{Triples.}
By definition of $(\min,+)$-product, for every $(i,j)$ we have $C[i,j]=\min_{\colindex\in[n^\alpha]}\{A[i,\colindex]+B[\colindex,j]\}$.
Thus, we associate with each entry $C[i,j]$ the triples $(i,1,j),(i,2,j),\dots,(i,n^\alpha,j)$ which refer to the $n^\alpha$ different indices considered in the $(\min,+)$-product definition.
The value of $C[i,j]$ is $A[i,\colindex]+B[\colindex,j]$ for at least one value $\colindex\in[n^\alpha]$; for such values $\colindex$ we say that $(i,\colindex,j)$ is \emph{relevant}.

Let $\hat A^r$ ($\hat B^r$) be the matrix $A^r$ ($B^r$) after replacing each entry whose absolute value is larger than $48\Delta$  with $\infty$, thereby effectively allowing a $(\min,+)$-product to ignore those entries.
Let $P^r = \hat A^r\star \hat B^r$ and define $\hat C^r[i,j]=P^r[i,j]+\tilde C[i,j^r]-\tilde C[i^r,j^r]+\tilde C[i^r,j]$.
Notice that $\hat C^r[i,j]\ge C[i,j]$.
For triples $(i,\colindex,j)$ where both $|A^r[i,\colindex]|\le 48\Delta$ and $|B^r[\colindex,j]|\le 48\Delta$ we have $\hat C^r[i,j]\le A[i,\colindex]+B[\colindex,j]$.
We call such triples \emph{covered} by the $r$th sample.
Finally, $\hat C$ is defined as the entry-wise minimum of $\hat C^r$ for $r=1,2,\dots,\rho$.
One  can compute each $P^r$, and eventually $\hat C$, using \cref{lem:bounded-matrix-multiplication}.
However, we employ a different approach which introduces additional speed-up.

\subsubsection*{Bounding the number of relevant and not covered triples.}
Our goal is to show that with high probability, most of the relevant triples $(i,\colindex,j)$ are covered by at least one of the $\rho$ samples.
Notice that for any pair $(i,j)$ we have $\hat C[i,j]= C[i,j]$ if and only if there exist $\colindex\in[n^\alpha]$ and $1\le r\le \rho$ such that $(i,\colindex,j)$ is relevant and is covered by the $r$th sample.
Since the relevant triples are unknown, we use a weaker definition of relevance and also a weaker definition of being covered:

\begin{definition}[{cf.~\cite[{Definition 2.2}]{BGSW19}}]
	We call a triple $(i,\colindex,j) \in [n]\times [n^\alpha]\times[n]$
	\begin{itemize}
		\item strongly relevant if $A[i,\colindex]+B[\colindex,j]=C[i,j]$;
		\item weakly relevant if $|A[i,\colindex]+B[\colindex,j]-C[i,j]|\le 16\Delta$;\footnote{\vphantom{$2^{2^2}$}We actually can use smaller constants, since by \cref{lem:approx-tilde-C} the additive error is at most $2\Delta$, while in~\cite{BGSW19} the additive error in $\tilde C$ is bounded by $4\Delta$. However, we prefer to use the original constants not to modify the claims and the proofs much.}
		\item strongly $r$-uncovered if for all $1\le r'\le r$ we have $|A^{r'}[i,\colindex]|>48\Delta$ or $|B^{r'}[\colindex,j]|>48\Delta$;
		\item weakly $r$-uncovered if for all $1\le r'\le r$ we have $|A^{r'}[i,\colindex]|>40\Delta$ or $|B^{r'}[\colindex,j]|>40\Delta$.
		
	\end{itemize}
	A triple is called strongly (weakly) uncovered if it is strongly (weakly) $\rho$-uncovered. Finally, a triple is strongly (weakly) $r$-covered if it is not strongly (weakly) $r$-uncovered.
\end{definition}

The following lemma is a straightforward generalization of~\cite[{Lemma 3.1}]{BGSW19} to our case.

\begin{restatable}[{Based on~\cite[{Lemma 3.1}]{BGSW19}}]{lemma}{lemboundonbadtriples}\label{lem:bound-on-bad-triples}
	With high probability for any $1\le r\le\rho$, the number of weakly relevant triples that are also weakly $r$-uncovered  is $\tilde O(n^{1.5+\alpha}+{n^{2+\alpha}/r^{1/3}})$.
\end{restatable}
\begin{proof}[sketch]
	The proof of~\cite[{Lemma 3.1}]{BGSW19} shows that for every $\colindex$, with high probability, the number of weakly relevant triples $(i,\colindex,j)$ that are also weakly $r$-uncovered is $\Ohtilde(n^{1.5}+n^2/r^{1/3})$.
	Since in our case there are only $n^\alpha$ possible values of $\colindex$, then, with high probability, the number of weakly relevant triples that are also weakly $r$-uncovered (among all possible values of $\colindex$) is $\Ohtilde(n^\alpha(n^{1.5}+ n^2/r^{1/3}))=\Ohtilde(n^{1.5+\alpha}+n^{2+\alpha}/r^{1/3})$.
\end{proof}

\subsubsection*{Efficient method for covering relevant triples.}
Instead of computing the $(\min,+)$-product $P^r=\hat A^r\star \hat B^r$ in the $r$th round, we follow the method of~\cite{BGSW19} and use their additional insight to reduce the runtime of the algorithm.
The proof of~\cite[{Lemma 3.1}]{BGSW19} has the property that for any round $r$, the triples that are counted as covered by the $r$th round are exactly triples of the form  $(i,\colindex,j)$ such that the triple $(i^r,\colindex,j^r)$  is weakly relevant and weakly $(r-1)$-uncovered.
Therefore, in order to cover the same triples at round $r$, after the algorithm chooses $i^r$ and $j^r$, the algorithm computes $L_r$ which is the set of $\colindex$ values such that $(i^r,\colindex,j^r)$ is both weakly relevant and weakly $(r-1)$-uncovered.
The algorithm removes all the columns $\colindex\notin L_r$ from $\hat A^r$ and all the rows $\colindex\notin L_r$ from $\hat B^r$.
Let $s_r=|L_r|$ be the number of surviving $\colindex$'s in round $r$.
The algorithm computes $\hat P^r=\hat A^r\star \hat B^r$ using \cref{lem:bounded-matrix-multiplication}, in  $O(\Delta\cdot M(n,s_r,n))$ time.

\subsubsection*{Runtime.}
For implementing the filtering in the $r$th iteration the algorithm checks $n^\alpha$ triples. For each triple the test takes $O(r)$ time (to check that the triple is $(r-1)$-uncovered).
Thus, the total time of the filtering takes in total $O(n^\alpha\cdot\rho^2)\le O(\rho n^{1+\alpha})$ time.
The runtime of computing the matrix products is $\sum_{r=1}^\rho \Ohtilde(\Delta M(n,s_r,n))$ time, which due to the analysis made in~\cite[{Lemma 3.2}]{BGSW19} is bounded by $\Ohtilde(\rho\Delta\cdot M(n,\frac{n^\alpha}{\rho^{1/3}},n))$ time with high probability.
Thus, the total running time of the second phase is $\tilde O(\rho\Delta\cdot M(n,\frac{n^\alpha}{\rho^{1/3}},n) + \rho\cdot n^{1+\alpha})$ time with high probability.

\subsection{Phase 3: complete the relevant uncovered triples.}
In the third phase the goal is to find all the triples $(i,\colindex,j)$ that are both strongly relevant and strongly uncovered, and use each such triple to update $\hat C$ by setting $\hat C[i,j]=\min\{\hat C[i,j],A[i,\colindex]+B[\colindex,j]\}$.
In order to identify all of these triples, we introduce the notions of approximately relevant and approximately covered.
\begin{definition}[{cf.~\cite[Definiton 2.6]{BGSW19}}]\label{def:approx-relevant-uncovered}
	We call a triple $(i,\colindex,j)\in I(i')\times[n^\alpha]\times I(j')$
	\begin{itemize}
		\item approximately relevant if $|A[i',\colindex]+B[\colindex,j']-\tilde C[i',j']|\le 8\Delta$, and
		\item approximately $r$-uncovered if for all $1\le r'\le r$ we have either $|A^{r'}[i',\colindex]|>44\Delta$ or $|B^{r'}[\colindex,j']|>44\Delta$ (or both).
	\end{itemize}
	A triple is called approximately uncovered if it is approximately $\rho$-uncovered.
	A triple which is both approximately relevant and approximately uncovered is called \emph{interesting}.
\end{definition}

A triple $(i',\colindex,j')$ where $i'$ and $j'$ are divisible by $\Delta$  and $\colindex\in[n^\alpha]$ is called a \emph{representative triple}.
A useful property of these definitions is that for a representative triple  $(i',\colindex,j')$, all the triples $(i,\colindex,j)\in I(i')\times\{\colindex\}\times I(j')$ (and, in particular, the representative triple itself) are similar in the sense that either all of these triples are interesting, or all of these triples are not interesting.
Thus, in order to find all the interesting triples, it suffices to test for every representative triple  if the triple is interesting.

The following lemma describes the relationships between the definitions of strongly, approximately and weakly relevant (uncovered) triples.
\begin{lemma}[{\cite[{Lemma 2.7}]{BGSW19}}]\label{lem:strong-approx-weak-rekation}
	Any strongly relevant triple is also approximately relevant.
	Any approximately relevant triple is also weakly relevant.
	The same holds also with ``relevant'' replaced by ``$r$-uncovered''.
\end{lemma}

To test whether a given (representative) triple is interesting, the algorithm tests if the triple is both approximately relevant and approximately uncovered.
The test whether a given triple is approximately relevant takes $O(1)$ time (see \cref{def:approx-relevant-uncovered}).
In order to find all approximately uncovered representative triples, the algorithm employs rectangular Boolean matrix multiplication as follows.
For each $\colindex\in[n^\alpha]$ let $U^{\colindex}$ be a matrix of size $\frac n\Delta \times \rho$ such that $U^{\colindex}[x,r]=1$ if and only if $|A^r[x\Delta,\colindex]|\le 44\Delta$ (see \cref{def:approx-relevant-uncovered}).
Similarly, let $V^{\colindex}$ be a matrix of size $\rho\times \frac n\Delta$ such that $V^{\colindex}[r,y]=1$ if and only if $|B^r[\colindex,y\Delta]|\le 44\Delta$.
Let $Z^{\colindex}$ be the Boolean matrix product $U^{\colindex}\cdot V^{\colindex}$.
Then, by \cref{def:approx-relevant-uncovered}, for every $i'=x\Delta$, $j'=y\Delta$ and $\colindex\in n^\alpha$, we have $(i',\colindex,j')$ is approximately uncovered if and only if $Z^{\colindex}[\frac{i'}\Delta,\frac{j'}\Delta]=1$.
Thus, the algorithm uses $n^\alpha$ rectangular Boolean matrix multiplications, in order to find all the interesting representative triples.
At the last step of the third phase, for any $(i',\colindex,j')$ which is known to be interesting from the previous step, the algorithm iterates over all $(i,\colindex,j)\in I(i')\times\{\colindex\}\times I(j')$ and updates $\hat C[i,j]\gets\min\{\hat C[i,j],A[i,\colindex]+B[\colindex,j]\}$.

\subsubsection*{Runtime.} 
At the first step of the phase, the algorithm finds all the interesting representative triples.
The time cost of finding all the approximately relevant representative triples is $O(\frac {n^{2+\alpha}}{\Delta^2})$.
The time cost of finding all the approximately uncovered representative triples is $O(n^\alpha\cdot M(\frac n\Delta,\rho,\frac n\Delta))$ time.
Then, in the second step of the phase the algorithm iterates over all the interesting triples, in time linear in the number of such triples which by \cref{lem:bound-on-bad-triples,lem:strong-approx-weak-rekation} is, with high probability, $\tilde O(n^{1.5+\alpha}+{n^{2+\alpha}/\rho^{1/3}})$.
In total, the third phase takes $O(n^\alpha\cdot M(\frac n\Delta,\rho,\frac n\Delta)+n^{1.5+\alpha}+n^{2+\alpha}/\rho^{1/3})$ time with high probability.

\subsection{Derandomization.}\label{sec:derandomize-min-plus}

The only randomized part of the algorithm described above is the process of choosing $i^r$ and $j^r$ from $[n]$.
The derandomization technique introduced by Bringmann et al.~\cite[Section 3.4]{BGSW19} is based on the following observations.

\subsubsection*{Choosing only indices which are divisible by {\boldmath$\Delta$}.}
Bringmann et al.~\cite[Section 3.4]{BGSW19} show that the algorithm has guarantees similar to~\cref{lem:bound-on-bad-triples} even if $i^r$ and $j^r$ are chosen uniformly at random from the set of indices in $[n]$ which are divisible by $\Delta$.
The following lemma states that after this change in the algorithm the number of interesting triples that will be covered in the third phase is still $\tilde O(n^{1.5+\alpha}+{n^{2+\alpha}/\rho^{1/3}})$.
Hence, the running time of the third phase will not increase.
\begin{lemma}[{\cite[{Section 3.4}]{BGSW19}}]\label{lem:derandomize-bound-on-bad-triples}
	With high probability, for any $1\le r\le\rho$ the number of approximately relevant, approximately $r$-uncovered triples is $\tilde O(n^{1.5+\alpha}+{n^{2+\alpha}/r^{1/3}})$.
\end{lemma}

\subsubsection*{Deterministic greedy choice.}
Let $r$ be an iteration of the algorithm and let $\colindex\in[n^\alpha]$.
Bringmann et al.~\cite{BGSW19} showed that one could count for every pair $(i',j')$, where $i'$ and $j'$ are divisible by $\Delta$, the number of triples $(i,\colindex,j)$ which are approximately relevant and approximately $(r-1)$-uncovered but will be covered by choosing $(i^r,j^r)\leftarrow (i',j')$.
The time cost of computing the counts for all pairs $(i',j')$ is $O((\frac n\Delta)^\omega)$ time for fixed $\colindex$ and~$r$.
Thus, a greedy algorithm that computes in the $r$th iteration the pair $(i', j')$ which maximizes the number of triples that will be covered, suffices to guarantee results which are at least as good as the expectation of the random process.
The additional time in total for the derandomization of the second phase is $O(\rho n^\alpha (\frac n\Delta)^\omega)$, and the guarantees of \cref{lem:derandomize-bound-on-bad-triples} hold.

\subsubsection{Total Runtime}
The additional runtime introduced by the derandomization process is exactly $O((\frac n\Delta)^\omega\cdot n^\alpha)= O(\frac {n^{\omega+\alpha}}{\Delta^\omega})$.
Summing up the running time of all the phases, the total runtime of the algorithm is:
\begin{align*}
\Ohtilde\left(\frac {n^{2+\alpha}}{\Delta^2}+n^2+\rho\Delta M(n,\frac {n^\alpha}{\rho^{1/3}},n)+n^\alpha\cdot M(\frac n\Delta,\rho,\frac n\Delta)+n^{1.5+\alpha}+\frac{n^{2+\alpha}}{\rho^{1/3}} + \frac {n^{\omega+\alpha}}{\Delta^\omega} \right)
\\=
\Ohtilde\left(\rho\Delta M(n,\frac {n^\alpha}{\rho^{1/3}},n)+n^\alpha\cdot M(\frac n\Delta,\rho,\frac n\Delta)+\frac{n^{2+\alpha}}{\rho^{1/3}} +n^{1.5+\alpha} + \frac {n^{\omega+\alpha}}{\Delta^\omega} \right).
\end{align*}

We use the results of Le Gall and Urrutia~\cite{LGU18}.
We set $\Delta= n^\delta$ and $\rho=n^s$ for some $s,\delta$ positive integers.
We will focus on $\alpha=\frac12$.
Hence, the running time is:

\begin{align*}
\Ohtilde\left(n^{\delta+s} M(n,n^{1/2-s/3},n)+n^{1/2}\cdot M(n,n^{\frac s{1-\delta}},n)^{1-\delta} +n^{5/2-s/3} + n^{1/2 + (1-\delta)\omega} \right)
.
\end{align*}

We optimize the running time by choosing $\delta=0.1881$ and $s=0.220425$.
By interpolating the results of~\cite{LGU18} the total running time of the algorithm is $O(n^{2.426524})$.

\begin{corollary}\label{cor:deterministic-min-plus-product}
	Let  $0<\alpha\le \frac12$.
	Let $A$ and $B$ be  integer matrices of sizes $n\times n^\alpha$  and $n^\alpha\times n$, respectively, such that $A$ is column-BD and $B$ is row-BD.\@
	There exists a deterministic algorithm that computes $A\star B$ whose runtime is $\Ohtilde(n^{2.426524})$.
\end{corollary}

\subsection{Randomized Algorithm}\label{sec:random_minplus}
In a recent paper, Chi et al.~\cite{CDXZ22} introduced an efficient randomized algorithm for computing the  $(\min,+)$-product of two square matrices $A$ and $B$ of size $n\times n$, for a broad range of structured matrices.
In particular, their algorithm is applicable to our case, except for the fact that our matrices are \emph{rectangular} and they described their algorithm for \emph{squared} matrices.
Notice that one can partition each rectangular matrix of size $n\times n^\alpha$ or $n^\alpha\times n$ into $n^{1-\alpha}$ sub-matrices of size $n^\alpha\times n^\alpha$, and so computing the $(\min,+)$-product of the two original matrices is reduced to $n^{2(1-\alpha)}$ computations of $(\min,+)$-products of squared matrices whose size is $n^\alpha\times n^\alpha$.
Such an approach leads to an algorithm that costs $O(n^{2(1-\alpha)}n^{\alpha(3+\omega)/2})= O(n^{2-\alpha/2+\alpha\omega/2})= O(n^{2+(\omega-1)\alpha/2})$ time.
Nevertheless, in this section we analyze the runtime obtained by adjusting the algorithm of Chi et al.~\cite{CDXZ22} to our case in order to obtain a more efficient runtime.

We emphasize that the only change to the algorithm of~\cite{CDXZ22} is replacing computations on squared matrices with computations on rectangular matrices.
Thus, correctness follows immediately from \cite{CDXZ22}, and our discussion focuses only on the runtime cost of the algorithm for our case. 
For complete definitions and details, the reader is referred to Chi et al.~\cite{CDXZ22}.

\subsubsection*{Overview.}
In our setting, the matrix $A$ is column-BD and the matrix $B$ is row-BD.
As~\cite[{Section 2}]{CDXZ22} observes, one can reduce the computation of a min-plus product of two such matrices to the computation of a min-plus product of two matrices where the values in the second matrix are row monotone (the values never decrease when scanning a row from left to right), 
and all of the values in both matrices are integers from the range $[1..cn]$ for some constant $c\in\mathbb N$.
Thus, for the rest of this section we assume that $B$ is row monotone and all the entries in $A$ and $B$ are integers from the range $[1..cn]$.

Notice that \cite{CDXZ22} contains two algorithm, the first one is a basic algorithm which provides the main ideas of the paper, and the second uses recursion to speedup the basic algorithm.
Here, we refer directly to the recursive algorithm.
Our algorithm follows the algorithm of \cite{CDXZ22} by replacing every part of the algorithm that computes a product of square matrices by an algorithm that computes a product of rectangular matrices.

The algorithm of Chi et al.~\cite{CDXZ22} is composed of two parts that both depend on a random prime $p\in[40n^\beta..80n^\beta]$ (for some parameter $\beta$ to be determined\footnote{Chi et al.~\cite{CDXZ22} uses '$\alpha$' to denote this parameter. We use '$\beta$' instead since '$\alpha$' is already taken.}).
In the first part the algorithm computes a matrix $C^*$ such that for every $i,j\in[1..n]$ we have $C^*_{i,j}=\floor{C_{i,j}/p}$ for some prime number $p$.
In the second part, the algorithm computes for every $i,j\in[1..n]$ the value $C_{i,j}\modulo p$.
Using the fact that $C_{i,j}=\floor{C_{i,j}/p}+C_{i,j}\modulo p$, after the execution the two parts, the algorithm completes the computation of $C$.
We describe the two parts below.

\subsubsection*{First part.} 
The first part of the algorithm computes the matrix $C^*$ defined as  $C^*_{i,j}=\floor{C_{i,j}/p}$. 
Notice that it is straightforward to compute $C^*$ in $O(n^3)$ time for squared matrices.
The algorithm of \cite{CDXZ22} applies \emph{segment tree}s \cite{BCKO08} to speedup the computation by a factor of $\tilde\Theta(n^\beta)$\footnote{The notation $\tilde\Theta(\cdot)$ suppresses polylogarithmic factors.} and so the runtime becomes $\Ohtilde(n^{3-\beta})$  (see \cite[{Lemma 3.2}]{CDXZ22}).
In our case, the straightforward algorithm for computing $C^*$ runs in $O(n^{2+\alpha})$ time, and the application of segments trees in our case provides the same speedup, so the computation of $C^*$ in our case becomes $\Ohtilde(n^{2+\alpha-\beta})$.

\subsubsection*{Second part.}
For $h=\ceil{\log p}+1$, the algorithm iteratively computes a sequence of matrices $C^{(h)},C^{(h-1)},\dots,C^{(0)}$ where for $\ell=h,h-1,\dots,0$ the entry $C^{(\ell)}_{i,j}$ is an approximation of $\floor{(C_{i,j}\modulo p)/2^\ell}$.
One important property of the sequence of matrices is that $C^{(0)}_{i,j} = C_{i,j}\modulo p$. 
Each matrix $C^{(\ell)}$ is iteratively computed using $C^{(\ell+1)}$.

In order to describe the computation of the different matrices $C^{(\ell)}$, Chi et al.~\cite{CDXZ22} define segments (unrelated to \emph{segment trees}) as follows.
Let $\ell\in [0..h]$, $B^*_{i,j}=\floor{B_{i,j}/p}$ and $B_{i,j}^{(\ell)}=\floor{(B_{i,j}\modulo p)/2^{\ell}}$.
A \emph{segment} $(i,k,[j_0,j_1])$ with respect to $\ell$ is a triple such that for every $j_0\le j\le j_1$ we have  $B_{k,j}^{(\ell)}=B_{k,j_0}^{(\ell)}$, $B_{k,j}^{*}=B_{k,j_0}^{*}$, $C_{k,j}^{(\ell)}=C_{k,j_0}^{(\ell)}$, and  $C_{k,j}^{*}=C_{k,j_0}^{*}$.
Chi et al.~\cite{CDXZ22} showed that for every $\ell\in[0..h]$ there are $O(n^3/2^\ell)$ (maximal) segments.
In our case, similar analysis shows that for every $\ell\in[0..h]$ there are $O(n^{2+\alpha}/2^\ell)$ segments.
Throughout the iterations, for every $\ell$ the algorithm of Chi et al.~\cite{CDXZ22} computes $O(1)$ sets of segments $\tlb$ (for $b\in[-10..10]$), used for computing the matrix $C^{(\ell)}$.
An important property proved by \cite[{Lemma 3.7}]{CDXZ22} is that for every $\ell$ and $b$ we have $\mathbb E[|\tlb|]=O(n^{3-\beta})$.
In particular, \cite[{Lemma 3.7}]{CDXZ22} proved that the probability of  every segment in level $\ell$ to be in $\tlb$ is $\Ohtilde(2^{\ell}/n^\beta)$.
In our settings this property still holds.
Therefore, in our case $\mathbb E[|\tlb|]=O(\frac {n^{2+\alpha}}{2^\ell})\cdot \Ohtilde(\frac{2^\ell}{n^\beta})=O(n^{2+\alpha-\beta})$.

Each iteration of the algorithm of Chi et al.~\cite{CDXZ22} is composed of three phases.
In the first phase, the algorithm computes a product of two polynomial matrices, by using fast matrix multiplication, where the entries of each matrix are polynomials of bounded degree $O(p)=O(n^\beta)$.
The runtime of this phase is $O(n^{\omega+\beta})$. In our case, the two matrices are of sizes $n\times n^\alpha$ and $n^\alpha\times n$, and the polynomials are also of bounded  degree $O(p)=O(n^\beta)$.
Therefore, by using fast \emph{rectangular} matrix multiplication, the running time of this phase in our case is generalized to $O(n^{n^\beta\cdot M(n,n^\alpha,n)})$ time. 
In the second phase, the algorithm of \cite{CDXZ22} computes  $C^{(\ell)}$, using the result of the first phase and the sets $T_b^{(\ell+1)}$.
The running time of this computation is dominated by two parts.
The first part is reading and manipulating the result matrix of the first phase.
The size of the matrix is $n\times n$ and each entry contains a polynomial of bounded degree $O(n^\beta)$. Thus, this part costs $O(n^{2+\beta})$ time, both in \cite{CDXZ22} and in our case.
The second part is dominated by $\sum_b|\tlb|$, which in \cite{CDXZ22} is $\Ohtilde(n^{3-\beta})$ and in our case is $\Ohtilde(n^{2+\alpha-\beta})$ in expectation.
In the third phase of the $\ell$th iteration, the algorithm computes the sets $\tlb$ using the sets $T_b^{(\ell+1)}$ and the matrix $C^{(\ell)}$. The running time of this phase is $\Ohtilde(\sum_b|T_b^{(\ell+1)}|)$, which in our case is $\Ohtilde(n^{2+\alpha-\beta})$ in expectation.

Thus, the total time cost of the second part of the algorithm of \cite{CDXZ22} in our case is  $\Ohtilde(n^\beta\cdot M(n,n^\alpha,n)+n^{2+\beta}+ n^{2+\alpha-\beta})= \Ohtilde(n^\beta\cdot M(n,n^\alpha,n)+ n^{2+\alpha-\beta})$.
For $\alpha= 1/2$ we optimize the running time by choosing $\beta=0.227908$. By interpolating the results of~\cite{LGU18} the total running time of the algorithm is $\Ohtilde(n^{2.272092})$ in expectation.
We mention that adapting the algorithm to compute $C=A\star B$ in $\Ohtilde(n^{2.272092})$ time with high probability is straightforward.

\begin{corollary}\label{cor:randomized-min-plus-product}
	Let  $0<\alpha\le \frac12$.
	Let $A$ and $B$ be  integer matrices of sizes $n\times n^\alpha$  and $n^\alpha\times n$, respectively, such that $A$ is column-BD and $B$ is row-BD.\@
	Then there exists a randomized algorithm that computes $A\star B$ in time $\Ohtilde(n^{2.272092})$ with high probability.
\end{corollary}

\cref{thm:rectangular-min-plus} follows from \cref{cor:randomized-min-plus-product,cor:deterministic-min-plus-product}.

\bibliography{bib}

\end{document}